\documentclass{amsarta}

\usepackage[margin=2.5cm,a4paper]{geometry}

\usepackage{amsmath}
\usepackage[T1]{fontenc}
\usepackage{wrapfig}

\newcommand{\algofont}[1]{\textnormal{\selectfont\sffamily#1}}

\usepackage[ruled]{algorithm}
\usepackage[noend]{algpseudocode}
\usepackage{enumerate,hyperref,verbatim}
\DeclareMathOperator{\eval}{val}
\DeclareMathOperator{\weight}{w}
\newcommand{\makeset}[2]{\ensuremath{ \{ #1 \: | \: #2 \} }}
\newcommand{\letters}{\ensuremath{ \Sigma }}

\newcommand{\mycount}{\textnormal{\textsl{count}}}
\newcommand{\mychoice}{\textnormal{\textsl{choice}}}

\newcommand{\mytext}{\ensuremath{t}}
\newcommand{\mytexti}{\ensuremath{t'}}
\newcommand{\pattern}{\ensuremath{p}}
\newcommand{\patterni}{\ensuremath{p'}}
\newcommand{\textc}{\ensuremath{X_{n+m}}}
\newcommand{\leftl}[1][(a)]{\ensuremath{\textnormal{\textsl{left}}#1}}
\newcommand{\rightl}[1][(a)]{\ensuremath{\textnormal{\textsl{right}}#1}}
\newcommand{\patternc}{\ensuremath{X_{m}}}

\newcommand{\allXc}{\ensuremath{X_1, \ldots, X_{n+m}}}

\newcommand{\algpop}{\algofont{Pop}}

\newcommand{\algFCPM}{\algofont{FCPM}}
\newcommand{\algfixdiff}{\algofont{FixEndsDifferent}}
\newcommand{\algfixsim}{\algofont{FixEndsSame}}
\newcommand{\algSPM}{\algofont{SimplePatternMatching}}
\newcommand{\algSET}{\algofont{SimpleEqualityTesting}}
\newcommand{\algET}{\algofont{EqualityTesting}}
\newcommand{\algradix}{\algofont{RadixSort}}
\newcommand{\algpairsncr}{\algofont{PairCompNcr}}
\newcommand{\algpairs}{\algofont{PairComp}}
\newcommand{\algpairscr}{\algofont{PairComp}}
\newcommand{\algremblocks}{\algofont{RemCrBlocks}}
\newcommand{\alggreedypairs}{\algofont{GreedyPairs}}

\newcommand{\PC}{{PC}}

\newcommand{\Fraref}[1]{(\hyperref[partition #1]{Fra~\ref{partition #1}})}
\newcommand{\Frarefall}{(\hyperref[partition 1]{Fra~1--Fra~4})}

\newtheorem{theorem}{Theorem}
\newtheorem{lemma}{Lemma}
\theoremstyle{remark}
\newtheorem{remark}{Remark}
\theoremstyle{definition}
\newtheorem{definition}{Definition}

\usepackage{color}

\providecommand{\Ocomp}{\mathcal{O}}

\newcommand{\twodots}{\mathinner{\ldotp\ldotp}}

\newcommand{\SLPrefall}{\eqref{form}}

\newtheorem{clm}{Claim}

\usepackage{color}
\definecolor{myYellow}{rgb}{0.9,0.9,0}

\begin{document}

\title{Faster fully compressed pattern matching by recompression}

\author{Artur Je\.z	
}

\address{
Max Planck Institute f\"ur Informatik, Campus E1 4,  DE-66123 Saarbr\"ucken, Germany
\and
Institute of Computer Science, University of Wroc{\l}aw, ul.\ Joliot-Curie~15, 50-383 Wroc{\l}aw, Poland\\
\texttt{aje@cs.uni.wroc.pl}
}

\thanks{Supported by NCN grant number 2011/01/D/ST6/07164, 2011--2014.}

\subjclass{F.2.2 Nonnumerical Algorithms and Problems}
\keywords{Pattern matching, Compressed pattern matching, Algorithms for compressed data, Straight-line programms, Lempel-Ziv compression}

\maketitle

\begin{abstract}
In this paper, a fully compressed pattern matching problem is studied.
The compression is represented by straight-line programs (SLPs), i.e.\ a
context-free grammars generating exactly one string;
the term fully means that \emph{both} the pattern \emph{and} the text are
given in the compressed form.
The problem is approached using a recently developed technique of local recompression:
the SLPs are refactored, so that substrings of the pattern and text
are encoded in both SLPs in the same way.
To this end, the SLPs are \emph{locally decompressed} 
and then \emph{recompressed} in a uniform way.

This technique yields an $\Ocomp((n+m)\log M)$ algorithm
for compressed pattern matching, assuming that $M$ fits in $\Ocomp(1)$ machine words,
where $n$ ($m$) is the size of the compressed
representation of the text (pattern, respectively),
while $M$ is the size of the decompressed pattern.
If only $m+n$ fits in $\Ocomp(1)$ machine words, the running time increases to $\Ocomp((n+m)\log M \log(n+m))$.
The previous best algorithm due to Lifshits had $\Ocomp(n^2m)$ running time.
\end{abstract}
\keywords{Pattern matching, Compressed pattern matching, Straight-line programms, Lempel-Ziv compression, Algorithms for compressed data}

\section{Introduction}
\label{sec:intro}
\subsubsection*{Compression and Straight-Line Programms}
Due to ever-increasing amount of data, compression methods
are widely applied in order to decrease the data's size.
Still, the stored data is accessed and processed.
Decompressing it on each such an occasion
basically wastes the gain of reduced storage size.
Thus there is a large demand for algorithms dealing directly with the compressed data,
without the explicit decompression.

Processing compressed data is not as hopeless, as it may seem:
it is a popular outlook, that compression basically extracts the hidden
structure of the text and if the compression rate is high,
the data has a lot of internal structure.
And it is natural to assume that such a structure
helps devising methods dealing directly
with the compressed representation.
Indeed, efficient algorithms for fundamental text operations
(pattern matching, equality testing, etc.) are known
for various practically used compression methods (LZ, LZW, their variants, etc.)%
~\cite{GawryLZW,GawryLZ,GawryLZWmany,GawryLZWfull,RytterSWAT,Rytter96,RytterfullLZW,Takeda00,PlandowskiSLPequivalence}.

The compression standards differ in the main idea as well
as in details. Thus when devising algorithms for compressed data,
quite early one needs to focus on the exact compression method,
to which the algorithm is applied. The most practical (and challenging)
choice is one of the widely used standards, like LZW or LZ.
However, a different approach is also pursued: for some applications
(and most of theory-oriented considerations)
it would be useful to \emph{model} one of the
practical compression standard by a more mathematically well-founded and `clean' method.
This idea lays at the foundations of the notion of
\emph{Straight-Line Programms} (SLP),
which are simply context-free grammars generating exactly one string.
Other reasons of popularity of SLPs is that usually they compress well the input text~\cite{RePair,Sequitur}
and that they are closely related to the LZ compression standard:
each LZ compressed text can be converted into an equivalent SLP of size $\Ocomp( n \log( N / n))$
and in $\Ocomp(n \log(N/n))$ time~\cite{SLPaprox,SLPaprox2,SLPaproxSakamoto,grammar} (where $N$ is the size of the decompressed text),
while each SLP can be converted to an equivalent LZ of $\Ocomp(n)$ size in linear time.
Lastly, a the greedy grammar compression can efficiently implemented and thus can be used
as a preprocessing to lsower compression methods, like those based on Burrows-Wheeler transform~\cite{grammarpreprocessing}.

\subsubsection*{Problem statement}
The problem considered in this paper is the \emph{fully compressed membership problem}
(FCPM),
i.e.\ we are given a text of length $N$ and pattern of length $M$,
represented by SLPs (i.e.\ context-free grammars in Chomsky normal form generating exactly one string)
of size $n$ and $m$, respectively.
We are to answer, whether the pattern appears in the text and
give a compact representation of all such appearances in the text.

\subsubsection*{Previous and related results}
The first algorithmic result dealing with the SLPs is 
for the compressed equality testing,
i.e.\ the question whether two SLPs represent the same text.
This was solved by Plandowski in 1994~\cite{PlandowskiSLPequivalence},
with $\Ocomp(n^4)$ running time.
The first solution for FCPM by Karpi\'nski et~al.\
followed a year later~\cite{SLPpierwszePM}.
Next, a polynomial algorithm for
computing various combinatorial properties of SLP-generated texts,
in particular pattern matching, was given by G\k{a}sieniec et~al.%
~\cite{RytterSWAT}, the same authors presented also a faster randomised algorithm
for FCPM~\cite{Rytter96}.
In 1997 Miyazaki et~al.~\cite{Takeda97} constructed $\Ocomp(n^2m^2)$ algorithm for FCPM.
A faster $\Ocomp(mn)$ algorithm for a special sub-case
(restricting the form of SLPs) was given in 2000 by Hirao et~al.~\cite{Takeda00}.
Finally, in 2007, a state of the art $\Ocomp(mn^2)$ algorithm was given by
Lifshits~\cite{LifshitsMatching}.

Concerning related problems, fully compressed pattern matching
was also considered for LZW compressed strings~\cite{RytterfullLZW} 
and a linear-time algorithm was recently developed~\cite{GawryLZWfull}.
Apart from that there is a large body of work dealing with the compressed pattern matching,
i.e.\ when the pattern is given explicitly,
for practically used compression standards.
We recall those for LZ and LZW as those compression standards are related to SLPs:
for LZW a linear-time algorithm was recently given~\cite{GawryLZW}
and the case of multiple pattern was also studied~\cite{GawryLZWmany}, with running time $\Ocomp(n \log M + M)$
(alternatively: $\Ocomp(n + M^{1+\epsilon})$).
For the LZ-compressed text, for which the problem becomes substantially harder than in LZW case,
in 2011 an $\Ocomp(n  \log (N/n) + m)$ algorithm, which is in some sense optimal,
was proposed~\cite{GawryLZ}.

\subsubsection*{Our results and techniques}
We give an $\Ocomp((n+m)\log M)$ algorithm for FCPM, i.e.\ pattern matching problem in which both the text and the pattern
are supplied as SLPs.
It assumes that numbers of size $M$ can be manipulated in constant time.
When this is not allowed and only numbers of $\Ocomp(n+m)$ time can,
the running time increases to $\Ocomp((n+m)\log M \log(n+m))$.
Since $M \leq 2^m$ this outperforms in any case the previously-best $\Ocomp(mn^2)$ algorithm~\cite{LifshitsMatching}.

\begin{theorem}
\label{thm:main result}
Assuming that numbers of size $M$ can be manipulated in constant time,
algorithm \algFCPM{} returns a $\Ocomp(n+m)$ representation
of all pattern appearances, where $n$ ($m$) is the size of the SLP-compressed text 
(pattern, respectively) and $M$ is the size of the decompressed pattern.
It runs in $\Ocomp((n+m)\log M)$ time.

If numbers of size $n+m$ can be manipulated in constant time,
the running time and the representation size increase by a multiplicative $\log(n+m)$ factor.

This representation allows calculation of the number of pattern appearances,
and if $N$ fits in $\Ocomp(1)$ codewords, also the position of the first, last etc.\ pattern.
\end{theorem}

Our approach to the problem is different
than all previously applied for compressed pattern matching
(though it does relate to dynamic string equality testing considered by Mehlhorn et~al.~\cite{MehlhornSU97}
and its generalisation to pattern matching by~Alstrup~et~al.~\cite{AlstrupBrodalRauhe}).
We do \emph{not} consider any combinatorial properties of the encoded strings.
Instead, we analyse and change the way strings are described by the SLPs in the instance.
That is, we focus on the SLPs alone, ignoring any properties of the encoded strings.
Roughly speaking, our algorithm aims at having all the strings in the instance
compressed `in the same way'.
To achieve this goal, we decompress the SLPs.
Since the compressed text can be exponentially long, we do this \emph{locally}:
we introduce explicit letters into the right-hand sides of the productions.
Then, we recompress these explicit strings uniformly:
roughly, a fixed pair of letters $ab$ is replaced by a new letter $c$
in both the string and the pattern;
such a procedure is applied for every possible pair of letters.
The compression is performed within the rules of the grammar and often it is needed to modify the grammar
so that this is possible.
Since such pieces of text are compressed in the same way,
we can `forget' about the original substrings of the input
and treat the introduced nonterminals as atomic letters.
Such recompression shortens the pattern (and the text) significantly:
roughly one `round' of recompression
in which every pair of letters that was present at the beginning of the `round'
is compressed shortens the encoded strings by a constant factor.
The compression ends when pattern is reduced to one letter,
in which case the text is a simple SLP-like representation of all pattern appearances.

\begin{remark}
Notice, that in some sense we build an SLP for both the pattern and string in a bottom-up fashion:
pair compression of $ab$ to $c$ is in fact introducing a new nonterminal with a production $c \to ab$.
This justifies the name `recompression' used for the whole process.
This is explained in details later on.
\end{remark}

\paragraph{Similar techniques}
While application the idea of recompression
to pattern matching is new,
related approaches were previously employed: most notably
the idea of replacing short strings by a fresh letter and iterating
this procedure was used by Mehlhorn et~al.~\cite{MehlhornSU97}
in their work on data structure for equality testing for dynamic strings
In particular their method can be straightforwardly applied
to equality testing for SLPs, yielding a nearly cubic algorithm
(as observed by Gawrychowski~\cite{GawrySLPEquality}).
However, the inside technical details of the construction
makes extension to FCPM problematic:
while this method can be used to build `canonical' SLPs
for the text and the pattern, there is no apparent way to
control how these SLPs actually look like
and how do they encode the strings.
An improved implementation of a similar data structure by Alstrup~et~al.~\cite{AlstrupBrodalRauhe}
solves those problems and allows a pattern matching.
Due to different setting the transitions to SLPs is not straightforward
(for instance, the running time is proportional to the number of appearances),
but it yields a nearly cubic algorithm~\cite{GawrySLPEquality}.

In the area of compressed membership problems~\cite{SLPmatchingNFA}, from
which the recompression method emerged, recent work of Mathissen and Lohrey~\cite{LohreySLP}
already implemented the idea of replacing strings with fresh letters
as well as modifications of the instance so that such replacement is possible.
However, the replacement was not iterated, and the newly introduced
letters were not be further compressed.

Lastly, a somehow similar algorithm, in which replaces pairs and blocks,
was proposed by Sakamoto in connection with the (approximate) construction of the smallest grammar for the input text~\cite{SLPaproxSakamoto}.
His algorithm was inspired by the \algofont{RePair} algorithm~\cite{RePair},
which is a practical grammar-based compressor.
However, as the text in this case is presented explicitly,
the analysis is much simpler and in particular it does not introduce the technique of modification of the grammar
according to the applied compressions

\paragraph{Other applications of the technique}
A variant of recompression technique has been used in order to establish
the computational complexity of the fully compressed membership problem for NFAs%
~\cite{fullyNFA}.
This method can also be applied in the area of word equations,
yielding simpler proofs and faster algorithms of many classical results in the area, like {\sf PSPACE}
algorithm for solving word equations, double exponential bound on the size of the solution,
exponential bound on the exponent of periodicity, \textit{etc}.~\cite{wordequations}.
Furthermore, a more tuned algorithm and detailed analysis yields a first linear-time algorithm
for word equations with one variable (and arbitrary many appearances of it)~\cite{onevarlinear}.
Lastly, the method can be straight-forwardly applied to obtain a simple algorithm for construction
of the (aproximation of) smallest gramamr generating a given word~\cite{grammar}.

\subsubsection*{Computational model}
Our algorithm uses \algradix{} and we assume
that the machine word is of size $\Omega(\log(n + m))$.
\algradix{} can sort $n+m$ numbers of size $\Ocomp((n+m)^c)$ in time $\Ocomp(c(n+m))$.

We assume that the alphabet of the input is $\{1, 2, \ldots, (n+m)^c\}$ for some constant $c$.
This is not restrictive, as we can sort the letters of the input and replace them with consecutive numbers, starting with $1$,
in total $\Ocomp((n+m)\log(n+m))$ time.

The position of the first appearance of the pattern in the text might be
exponential in $n$, so we need to make some assumptions in order to be able to output such a position.
Assuming that $N$ fits in a constant amount of codewords,
our algorithm can also output the position of the first, last etc.\ position of the pattern.

We assume that the rules of the grammar are stored as lists, so that insertion and deletion of characters
can be done in constant time (assuming that a pointer to an element is provided).

\subsubsection*{Organisation of the paper}
As a toy example, we begin with showing that the recompression can be used to check the equality of two explicit strings,
see Section~\ref{sec:toy example}.
This introduces the first half of main idea of recompression: iterative replacement of pairs and blocks, as well as some key
ideas of the analysis.
On the other hand, it completely ignores the (also crucial) way the SLP is refactored to match the applied recompression.
In the next section it is explained, how this approach can be extended to pattern matching,
we again consider only the case in which the text and pattern are given explicitly.
While the main is relatively easy, the method and the proof involve a large case inspection.

Next, in Section~\ref{sec:outline} we show how to perform the equality testing in case of SLPs.
This section introduces the second crucial part of the technique: modification of SLP in the instance according to the compressions.
This section is independent form Section~\ref{sec: first and last} and can be read beforehand.
In the following section it is showed how to merge the results of Section~\ref{sec: first and last} and Section~\ref{sec:outline},
yielding and algorithm for fully compressed pattern matching.

In the last Section~\ref{sec:improvements} we comment how to improve the running time from
$\Ocomp((n+m) \log M \log (n+m))$ to $\Ocomp((n+m) \log M)$ when $M$ fits in $\Ocomp(1)$ machine words.

\section{Toy example: equality testing}
\label{sec:toy example}
In this section we introduce the recompression technique and apply it in the trivial case of equality testing
of two explicit strings, i.e.\ their representation is not compressed.
This serves as an easy introduction.
In the next section we take this process a step further, by explaining, how to perform a pattern
matching for explicit strings using recompression.
To stress the future connection with the pattern matching, we shall use letters \pattern{} and \mytext{}
to denote the two strings for which we test the equality.

In case of equality testing, our approach is similar to the one of Mehlhorn et~al.~\cite{MehlhornSU97}
from their work on equality testing for the dynamic strings.
In that setting, we were given a set of strings, initially empty, and a set of operations
that added new strings to the set.
We are to create a data structure that could answer whether two strings in this collection are equal or not.

The method proposed by Mehlhorn et~al.~\cite{MehlhornSU97} was based on iterative replacement of strings:
they defined a schema, which replaced a string $s$ with a string $s'$
(where $|s'| \leq c |s|$ for some constant $c < 1$) and iterated the process until a length-1 string was obtained.
Most importantly, the replacement is injective, i.e.\ if $s_1 \neq s_2$ then they are replaced with different strings%
\footnote{This is not a information-theory problem, as we replace only strings that appear in the instance and moreover can reuse original letters.}.
In this way, for each string we calculate its unique signature and two strings are equal if and only if their signatures are.

The second important property of this schema is that the replacement is `local': $s$ is partitioned into blocks of a constant amount
of letters and each of them is replaced independently.

The recompression, as presented in this section, is a variant of this approach,
in which a different replacement schema is applied.
To be more specific, our algorithm is based on two types of `compressions'
performed on strings:
\begin{description}
	\item[pair compression of $ab$]
	For two different letters $ab$ appearing in \pattern{} or \mytext{}
	replace each of $ab$ in \pattern{} and \mytext{} by a \emph{fresh} letter $c$.
	\item[$a$'s block compression]
	For each maximal block $a^\ell$, with $\ell >1$,
	that appears in \pattern, replace all $a^\ell$s in \pattern{} and \mytext{}
	by a fresh letter $a_\ell$.
\end{description}
By a fresh letter we denote any letter that does not appear in \pattern{} or \mytext.
We adopt the following notational convention throughout rest of the paper:
whenever we refer to a letter $a_\ell$, it means that the block compression
was done for $a$ and $a_\ell$ is the letter that replaced $a^\ell$.
The $a$-block $a^\ell$ is \emph{maximal}, when it cannot be extended by a letter $a$ to the left, nor to the right.

Clearly, both compressions preserve the equality of strings
\begin{lemma}
\label{lem:SET properly tests}
Let \patterni, \mytexti{} be obtained from \pattern{} and \mytext{}
by a pair compression (or block compression).
Then $\pattern = \mytext$ if and only if $\patterni = \mytexti$.
\end{lemma}

Using those two operations, we can define the algorithm for testing the equality of two strings
\begin{algorithm}[H]
	\caption{\algSET: outline}
	\label{alg:set}
	\begin{algorithmic}[1]
	\While{$|\pattern|>1$ and $|\mytext|>1$}
		\State $L \gets $ list of letters appearing in \mytext{} and \pattern
		\State $P \gets $ list pairs appearing in \mytext{} and \pattern
    \For{each $a \in L$} {} compress blocks of $a$
		\EndFor
    \For{each $ab\in P$} {} compress pair $ab$ 
		\EndFor
	\EndWhile
	\State Naively check the equality and output the answer.
	\end{algorithmic}
\end{algorithm}
We call one iteration of the main loop a \emph{phase}.

The crucial property of \algSET{} is that in each phase the lengths of \pattern{} and \mytext{}
shorten by a constant factor

\begin{lemma}
\label{lem:SET shortens}
When $|\pattern|, |\mytext| > 1$ then one phase shortens those lengths by a constant factor.
\end{lemma}
\begin{proof}
Consider two consecutive letters $a$ and $b$ of \pattern{} (the proof for \mytext{} is the same).
We claim that at least one of them is compressed in a phase.

\begin{clm}
\label{clm:two letters are compressed}
Consider any two consecutive letters in \pattern{} or \mytext{} at the beginning of the phase.
Then at least one of those letters is compressed till the end of the phase.
\end{clm}
\begin{proof}
If they are the same, then they are compressed during the blocks compression.
So suppose that they are different. Then $ab \in P$ and we try to compress this appearance during the pair compressions.
This fails if and only if one of letters from this appearance was already compressed when we considered $ab$ during the pair compression.
\qed
\end{proof}

So each uncompressed letter can be associated with a letter to the left and to the right, which were compressed
(the first and last letter can be only associated with a letter to the right/left, respectively).
Since when a substring is compressed, it is of length at least two,
this means that no compressed letter is associated with two uncompressed letters.
So, for a pattern \pattern{} there are at most $\frac{|\pattern|+2}{3}$ uncompressed
letters (the $+2$ comes from the first/last letter that can be uncompressed and do not have a compressed letter to the left/right)
and at least $\frac{2|\pattern|-2}{3}$ compressed ones.
Hence, the length of the pattern at the end of a phase is at most
\begin{equation*}
\frac{|\pattern|+2}{3} + \frac{1}{2} \cdot \frac{2|\pattern|- 2}{3} = \frac{2 |\pattern| + 1}{3} \leq \frac{5}{6} |\pattern| \enspace ,
\end{equation*}
where the last inequality holds in the interesting case of $|\pattern| > 1$.
\qed
\end{proof}

This shows that there are at most $\Ocomp(\log(\min(m,n)))$ phases, however, the running time can be in fact
bounded much better: one phase takes only linear time, assuming that the alphabet $\letters$ can be identified
with numbers from the set $\{1, 2, \ldots, (n+m)^c\}$ for some constant $c$.
Since the lengths of \pattern{} and \mytext{} shorten by a constant factor in each phase, this yields a total linear running time.

\begin{lemma}
\label{lem:SET linear}
Assuming that in the input $\letters$ can be identified with $\{1, 2, \ldots, (n+m)^c\}$,
one phase of \algSET{} can be implemented in $\Ocomp(|\pattern| + |\mytext|)$ time.
\end{lemma}
\begin{proof}
In order to show the running time bounds, we need to ensure that  letters in \pattern{} and \mytext{} form an interval of consecutive letters
(in this way \algradix{} can be applied).
We prove the technical claim after the main part of the proof.
\begin{clm}
\label{clm:consecutive values}
Without loss of generality, at the beginning of each phase the letters present in \pattern{}
and \mytext{} form an interval $\{k+1,k+2,\ldots, k+k'\}$ for some $k$ and $k' \leq |\pattern|+|\mytext|$.
Ensuring this takes at most $\Ocomp(|\pattern|+|\mytext|)$ time in a phase.
\end{clm}

Using Claim~\ref{clm:consecutive values} the proof of the lemma is nearly trivial:
We go through the \pattern{} and \mytext.
Whenever we spot a pair $ab$ (of different letters),
we create a record $(a,b,p)$, where $p$ is the pointer to this appearance of $ab$.
Similarly, when we spot a maximal block $a^\ell$ (where $\ell > 1$),
we put a record $(a,\ell,p)$, where $p$ is again a link to this maximal block (say, to the first letter of the block).
Clearly, this takes linear time.

We sort the triples for blocks, using \algradix{} (we ignore the third coordinate).
Since the letters form an interval of size at most $|\pattern| + |\mytext|$
and blocks have length at most $|\pattern| + |\mytext|$,
this can be done in $\Ocomp(|\pattern| + |\mytext|)$ time.
Then we go through the sorted list and replace $a^\ell$ with $a_\ell$, for $\ell > 1$.
Since all appearances of $a^\ell$ are consecutive on the sorted list,
this can be done in time $\Ocomp(1)$ per processed letter.
Hence, the total running time is linear.

Similarly, we sort the triples of pairs.
For each $ab$ on the list we replace all of its appearances by a fresh letter.
Note, that as the list is sorted, before considering a pair $a'b'$ we either replaced all,
or none appearances of a different pair $ab$ (depending on whether $ab$ is earlier or later in the list).
Hence this effectively implements iterated pair compression.

Note that it might be that a link to pair $ab$ is invalid, as one of letters $ab$ was already replaced.
In such a case we do nothing.

It is left to show the technical Claim~\ref{clm:consecutive values}:
\begin{proof}[of Claim~\ref{clm:consecutive values}]
We show the claim by an induction on the number of phase.

Consider the first phase.
We assumed that the input alphabet consists of letters that can be identified with subset of $\{1, \ldots, (n+m)^c\}$.
Treating them as vectors of length $c$ over $\{0, \ldots, (n+m) - 1\}$ we can sort them using \algradix{}
in $\Ocomp(c(n+m))$ time, i.e.\ linear one. Then we can re-number those letters to $1$, $2$, \ldots, $k$
for some $k \leq n+m$.
This takes $\Ocomp(n+m) = \Ocomp(|\pattern|+|\mytext|)$ time.

Suppose that at the beginning of the phase the letters formed an interval $[k+1\twodots k+k']$.
Each new letter, introduced in place of a compressed pair or block, is assigned a consecutive value,
starting from $k+k'+1$
and so after the phase the letters appearing in \pattern{} and \mytext{} are either within $[k+1 \twodots k+k']$ 9the old letters)
or within an interval $[k+k'+1\twodots k + k'']$ (the new letters),
for some $k' \leq k'' \leq k' + |\pattern| + |\mytext|$
(the second inequality follows from the fact that introduction of a new letter shortens \pattern{} or \mytext{}
by at least one letter).
It is now left to re-number the letters $[k+1\twodots k + k']$, so that only those appearing
in \pattern{} and \mytext{} have valid numbers: we go through \pattern{} \mytext{} and for each letter $a$ with number
in $[k\twodots k' + k]$ we increase the counter $\mycount[a]$ by $1$.
Then we go through $\mycount$ and assign consecutive numbers, starting from $k+k''+1$ to letters with non-zero count.
Lastly, we replace the values of those letters in \pattern{} and \mytext{} by the new values.
\qed 
\end{proof}
\qed
\end{proof}
By iterative application of Lemma~\ref{lem:SET properly tests} each compression performed by \algSET{}
preserves the equality of strings, so \algSET{} returns a proper answer.
Concerning the running time of one phase takes $\Ocomp(|\pattern| + |\mytext|)$, by Lemma~\ref{lem:SET linear},
and as $|\pattern|$ and $|\mytext|$ shorten by a constant factor in each phase, see Lemma~\ref{lem:SET shortens},
this takes in total $\Ocomp(n + m)$ time.

\begin{theorem}
\algSET{} runs in $\Ocomp(n + m)$ and tests the equality of \pattern{} and \mytext.
\end{theorem}

\subsubsection*{Building of a grammar}
Observe, that as already noted in the Introduction,
\algSET{} basically generates a context free grammar, whose some nonterminals generate \pattern{} and \mytext{}
(additionally, this context free grammar is an SLP, which are formally defined in the later section).
To be more precise: observe that each replacement of $ab$ by $c$ corresponds to an introduction of a new nonterminal
$c$ with a production $c \to ab$ and replacement of each $ab$ with $c$, that generates the same string as $ab$ does.
Similarly, the replacement of $a^k$ with $a_k$ corresponds to an introduction of a new nonterminal $a_k$
with a rule $a_k \to a^k$.

\section{Toy example: pattern matching}
\label{sec: first and last}
The approach used in the previous section basically applies also to the pattern matching,
with one exception: we have to treat the `ends' of the pattern in a careful way.
Consider $\mytext = ababa$ and $\pattern = baba$. Then compression of $ab$ into $c$
results in $\mytexti = cca$ and pattern $\patterni = bca$, which no longer appears in \mytexti.
The other problem appears during the block compression: consider $\pattern = aab$
and $\mytext = aaab$. Then after the block compression the pattern is replaced with $\pattern' = a_2b$
and text with $\mytexti = a_3b$.

In general, the problems arise because the compression in \mytext{} is done partially
on the \pattern{} appearance and partially outside it, so it cannot be reflected in the compression of \pattern{} itself.
We say that the compression \emph{spoils pattern's beginning} (\emph{end})
when such partial compression appears on pattern appearance beginning (end, respectively).
In other words, when $a$, $b$ are the first and last letters of the \pattern,
then we cannot perform a pair compression for $ca$ or $bc$ (for any letter $c$),
nor the $a$ or $b$ block compression.

\begin{lemma}
\label{lem:when appearances are preserved}
If the pair compression (block compression) does not spoil the end, not the beginning,
then there is a one-to-one correspondence
\end{lemma}

In the first example, i.e.\ $\mytext = ababa$ and $\pattern = baba$,
spoiling of the pattern's beginning can be circumvented by
enforcing a compression of the pair $ba$ in the first place:
when two first letters of the pattern are replaced by a fresh letter $c$,
then the beginning of the pattern no longer can be spoiled in this phase
(as $c$ will not be compressed in this phase).
We say, that pattern's beginning (end) is \emph{fixed} by a pair or block compression,
if after this compression a first (last, respectively) letter of the pattern
is a fresh letter, so it is not in $L$ and no pair containing it is in $P$.

Our goal is to fix both the beginning and end, without spoiling any of them.
Notice, that the same compression can at the same time fix the beginning and spoil the
end: for instance, for $\mytext = ababa$ and $\pattern = bab$,
compressing $ba$ into $c$ fixes the beginning and spoils the end while compression of $ab$ into $c$
spoils the beginning and fixes the end.
This example demonstrates that the case in which the first and last letter of the pattern are the same
is more problematic than the case in which they are different.

\begin{algorithm}[H]
	\caption{\algSPM: outline}
	\label{alg:example}
	\begin{algorithmic}[1]
	\While{$|\pattern|>1$}
		\State $L \gets $ list of letters in $\pattern$, $\mytext$ and $P \gets $ list of pairs in $\pattern$, $\mytext$
		\If{$\pattern[1] \neq \pattern[|\pattern|]$}
			\State $\algfixdiff(\pattern[1],\pattern[|\pattern|])$
		\Else \Comment{$\pattern[1] = \pattern[m]$}
			\State $\algfixsim(\pattern[1])$
		\EndIf
		\For{$a \in L$} {} compress blocks of $a$ in $\pattern$ and $\mytext$ \label{spm block compression}
		\EndFor
		\For{$ab \in P$} {} compress pair $ab$ in $\pattern$ and \mytext{} \label{spm pair compression}
		\EndFor
	\EndWhile
	\State check, if \pattern[1] appears in \mytext{}
	\end{algorithmic}
\end{algorithm}

There are four main subcases, when trying to fix the beginning, they depend on whether:
\begin{itemize}
	\item the first and last letter of the pattern are the same are not
	\item the first and second letter of the pattern are the same or not
	(i.e.\ whether \pattern{} begins with a pair or a block).
\end{itemize}
We consider them in the order of increasing difficulty.

Suppose that the first and last letter of the pattern are different.
If moreover, the first two letters of the pattern are $ab$ for $a \neq b$,
then we can fix the beginning by compressing the pair $ab$,
before any other pairs (or blocks) are compressed.
This will fix the beginning and not spoil the end (since the last letter is not $a$).
This cannot be applied, when $a = b$, or in other words,
\pattern{} has a leading $\ell$-block of letters $a$ for some $\ell > 1$.
The problem is that each $m$-block for $m \geq \ell$ can begin an appearance
of the pattern in the text.
The idea of the solution is to replace the leading $\ell$-block
of \pattern{} with $a_\ell$, but then treat $a_\ell$ as a `marker'
of a (potential) beginning of the pattern, meaning that each block
$a^m$ for $m\geq \ell$ should be replaced with a string ending with $a_\ell$.
To be more specific:
\begin{itemize}
	\item
	for $m \leq \ell$ each $m$-block is replaced by a fresh letter $a_m$;
	\item
	for $m > \ell$ each $m$-block is replaced by a \emph{pair} of letters $a_ma_\ell$, where $a_m$ is a fresh letter.
\end{itemize}
This modifies the block compression, however,
there is no reason, why we needed to replace $a^m$ by exactly one letter in the block compression,
two letters are fine, as long as:
\begin{itemize}
	\item the replacement function is injective;
	\item they are shorter than the replaced text;
	\item the introduced substring consists does not appear in \pattern{} and \mytext.
\end{itemize}
We shall not formalise this intuition, instead the proofs will simply show that there is a one-to-one
correspondence between appearances of the pattern before and after such modified block compression.

For instance, in the considered example $\mytext = aaab$ and $\pattern = aab$ we obtain
$\mytext = a_3a_2b$ and $\pattern = a_2b$;
clearly \pattern{} has an appearance in \mytext.
In this way we fixed the pattern beginning.

Now it is left to fix the pattern's end, which is done in the same way.
Note that we may need to compress a pair including a letter introduced
during the fixing of the beginning, but there is no additional difficulty in this
(though this somehow contradicts the earlier approach that we do not replace letters introduced in the current phase).

\begin{algorithm}[H]
	\caption{$\algfixdiff(a,a')$}
	\label{alg:fixdiff}
	\begin{algorithmic}[1]
			\State $b \gets \pattern[2]$
			\If{$a \neq b$}	\Comment{Compress the leading pair $ab$}
				\State comrpess $ab$ in $t$ and $p$
			\Else \Comment{$a = b$: compress the $a$ blocks}
				\State let $\ell \gets $ length of the $p$'s $a$-prefix
				\For{$m\leq \ell$}
					\State replace each maximal block $a^m$ in $p$, $t$ by $a_m$		
				\EndFor
				\For{$m>\ell$}
					\State replace each maximal block $a^m$ in $p$, $t$ by $a_ma_\ell$		
				\EndFor
				\If{\mytext{} ends with $a_\ell$} {} remove this $a_\ell$ \Comment{Cannot be used by pattern appearance anyway}
				\EndIf
			\EndIf
			\Comment{Symmetric for the ending letter}
		\end{algorithmic}
\end{algorithm}

The described approach does not work when the first and last letter of the pattern are the same.
As a workaround, we alter the pattern so that the first and last letter are in fact
different and then apply the previous approach.
The idea is to introduce the `markers' $a_L$ and $a_R$ which denote
the potential beginning and ending of the pattern; we assume that $a_L \neq a_R$, even if $a^\ell = a^r$.
They work as the marker $a_\ell$ in the block compression in the previous case:
Let $a^\ell$ and $a^r$ be the $a$-prefix and $a$-suffix of \pattern.
We replace the $a$-prefix ($a$-suffix) of the pattern with $a_L$ ($a_R$, respectively)
and then make a block compression for $a$, in which $a^m$, for $m \geq \ell,r$,
is replaced by $a_R a_m a_L$.
This reflects the fact that $a^m$ can \emph{both} begin and end the pattern appearance,
the former consumes ending $a_L$ and the latter the leading $a_R$.
The exact replacement of $a^m$ for $m \leq \max(\ell,r)$ depends on whether $\ell < r$, $\ell = r$ or $\ell > r$,
for instance, when $\ell = r$:
\begin{itemize}
	\item for $m < \ell$ we replace $m$-blocks with $a_m$;
	\item for $m = \ell$ we replace $\ell$-blocks with $a_Ra_L$;
	\item for $m > \ell$ we replace $\ell$-blocks with $a_Ra_ma_L$.
\end{itemize}
the other replacement schemes are similar.

Note that in this way it is possible that the \pattern{} begins with $a_R$ or ends with $a_L$,
none of which can be used by a pattern appearance.
For simplicity, we remove such $a_R$ and $a_L$.

For $\ell = r = 1$ this actually enlarges the LZ-representations
(and for $\ell = r = 2$ not always decreases the length).
To fix this make additional round of pair replacement, immediately after the blocks replacement:
we make the compression of pairs of the form $\makeset{a_Lb}{b \in \Sigma \setminus \{a_L\}}$
(note that those pairs cannot overlap, so all of them can be replaced in parallel),
followed by compression of pairs $\makeset{ba_R}{b \in \Sigma \setminus \{a_R\}}$.
The latter compression allows compression of the letters introduced in this phase,
i.e.\ $a_Lba_R$ is first compressed into $ca_R$ and then into $c'$.
It can be routinely checked, that this schema  shortens both \pattern{} and \mytext:
to see this observe that $bab'$ is first replaced with $ba_Ra_Lb'$
and then by $ba_Rc$ and finally with $c'c$, which is shorter than $bab'$;
other cases are analysed similarly.
When afterwards a block compression and pair compression is applied, Lemma~\ref{lem:SET shortens} still holds,
though with a larger constant.

\begin{lemma}
\label{lem:compression of ends different}
When the first and last letter of the pattern are different, in $\Ocomp(|\pattern| + |\mytext|)$
time we can fix both the beginning and end without prior spoiling them.

There is a one-to-one correspondence between the pattern appearances in the new text and old pattern appearances in the old text.
\end{lemma}
\begin{proof}
It was already described, how to perform the appropriate operations,
it is left to analyse their properties and implementations.
 
\paragraph{Fixing the beginning}
Let $a'$ be the second letter of \pattern.
Suppose first that $a' \neq a$. We (naively) perform the modified compression of the pair $aa'$,
by reading both \pattern{} and \mytext{} from the left to the right.
Note that the beginning and end were not spoiled in the process, and so there is a one-to-one correspondence
of new and old pattern appearances.

So suppose that $a' = a$, let $a^\ell$ be the $a$-prefix of \pattern.
Then we (naively) perform the compression of blocks for the letter $a$.
We show that no pattern appearance was lost, nor that any new pattern `appearance' was introduced.
So let $\mytext = w_1a^mw_2w_3$ and $\pattern = a^\ell w_2$, for $m > \ell$.
Observe that as the first and last letter of the pattern are different, we know that $w_2 \neq \epsilon$.
Let $w_i$ be replaced by $w_i'$.
Then the new text is $\mytexti = w_1'a_ma_\ell w_2'w_3'$ and the new pattern is $\patterni = a_\ell w_2'$,
thus there is a pattern appearance in the new text.
The case in which $m = \ell$ is shown in the same way.

Conversely, let $w_1'a_\ell w_2'w_3' $ be the new text and $a_\ell w_2'$ the new pattern.
Then the pattern was obtained from $a^\ell w_2$ for some $w_2$.
Furthermore, $w_1' a^\ell$ was obtained from some $w_1 a^m$ for $m \geq \ell$,
(this is the only way to obtain $a_\ell$),
also, the only way to obtain $w_2'$ is from the same $w_2$.
Hence, no new pattern appearance was introduced.

This fixes the pattern beginning and as the last letter of \pattern{} is not $a$, it did not spoil the pattern end.

\paragraph{Fixing the end}
We want to apply exactly the same procedure at the end of the \pattern.
However, there can be some perturbation, as fixing the beginning
might have influenced the end:
\begin{itemize}
	\item the last letter could have been already compressed,
	which can happen only when $b=a'$.
	In this case we got lucky and we make no additional compression,
	as the end of the pattern has been already fixed.
	\item the second last letter (say $b'$) of \pattern{} equals $a'$
	and it was compressed, into the letter $c$
	(either due to pair compression or block compression).
	In this case we make the compression of the pair $ca'$,
	even though $c$ is a fresh letter.
	Note, that as $c$ is the first letter of this pair,
	this will not spoil the beginning of the pattern.
\end{itemize}
The rest of cases, as well as the analysis of the above exceptions,
is the same as in the case of fixing the beginning.
\qed
\end{proof}

Now we consider the more involved case in which the first and last letter
of the pattern are the same.

\begin{algorithm}[H]
	\caption{\algfixsim}
	\label{alg:fixsim}
	\begin{algorithmic}[1]
			\State let $\ell \gets$ the length of $p$'s $a$-prefix, $r \gets$ the length of $a$-suffix
			\State replace the leading $a^\ell$ and ending $a^r$ in \pattern{} by $a_L$ and $a_R$
			\If{$\ell = r$}
				\For{$m < \ell$}
					\State replace each maximal $a^m$ in $p$, $t$ by $a_m$
				\EndFor
				\State replace each maximal $a^\ell$ in $p$, $t$ by $a_Ra_L$
				\For{$m > \ell$}
					\State replace each maximal $a^m$ in $p$, $t$ by $a_Ra_ma_L$
				\EndFor
			\EndIf
			\If{$\ell < r$}
				\For{$m < \ell$}
					\State replace each maximal $a^m$ in $p$, $t$ by $a_m$
				\EndFor
				\State replace each maximal $a^\ell$ in $p$, $t$ by $a_L$
				\For{$r > m > \ell$}
					\State replace each maximal $a^m$ in $p$, $t$ by $a_ma_L$
				\EndFor
				\For{$m \geq r$}
					\State replace each maximal $a^m$ in $p$, $t$ by $a_Ra_ma_L$
				\EndFor
			\EndIf
			\If{$\ell > r$}
				\For{$m < r$}
					\State replace each maximal $a^m$ in $p$, $t$ by $a_m$
				\EndFor
				\For{$r \leq m < \ell$}
					\State replace each maximal $a^m$ in $p$, $t$ by $a_Ra_m$
				\EndFor
				\State replace each maximal $a^r$ in $p$, $t$ by $a_Ra_L$
				\For{$m > \ell$}
					\State replace each maximal $a^m$ in $p$, $t$ by $a_Ra_ma_L$
				\EndFor
			\EndIf
			\If{\mytext{} ends with $a_L$} remove this $a_L$ \EndIf
			\If{\mytext{} begins with $a_R$} remove this $a_R$ \EndIf
			\State compress all pairs of the form $a_Lb$ with $b \in \Sigma \setminus \{ a_L \}$
			\State compress all pairs of the form $ba_R$ with $b \in \Sigma \setminus \{ a_R \}$
			\If{$1 = r < \ell$}
				\State compress all pairs of the form $a_1b$ with $b \in \Sigma \setminus \{ a_1 \}$
			\EndIf
	\end{algorithmic}
\end{algorithm}

\begin{lemma}
\label{lem:compression of ends same}
When the first and last letter of the pattern are equal,
in $\Ocomp(|\pattern| + |\mytext|)$ time we can fix both the beginning and end without prior spoiling them.

There is a one-to-one correspondence between the pattern appearances in the new text and old pattern appearances in the old text.
\end{lemma}
\begin{proof}
Let the first (and last) letter of the pattern be $a$.
There is a simple special case, when $\pattern \in a^*$.
Then it is enough to perform a usual compression of $a$ blocks and mark the letters $a_m$ for $m \geq \ell$,
as each such letter corresponds to $m-\ell+1$ appearances of the pattern.
To this end we perform the $a$-blocks compression (for blocks of $a$ only),
which includes the sorting of blocks according to their length.
Hence blocks of length at least $\ell$ can be identified and marked.

So consider now the case, in which the pattern has some letter other than $a$, i.e.\
$\pattern = a^\ell u a^r$ where $u \neq \epsilon$ and it does not begin, nor end with $a$.
The main principle of the replacement was already discussed,
i.e.\ first a tuned version of the $a$ blocks compression is performed,
which introduces markers $a_L$ and $a_R$ denoting the pattern beginning and end,
respectively; then a compression of the pairs of the form $\makeset{a_Lb}{b \in\Sigma \setminus \{a_L\}}$
and finally $\makeset{ba_R}{b \in\Sigma \setminus \{a_R\}}$.

While the block compression scheme was already given for $r = \ell$,
the ones for $\ell > r$ and $r < \ell$ were not, we start with their precise description,
see also Algorithm~\ref{alg:fixdiff}.

The replacement of blocks for $\ell < r$ is as follows:
\begin{itemize}
	\item for $m<\ell$ maximal blocks $a^m$ are replaced by $a_m$;
	\item for $m = \ell$ maximal blocks $a^\ell$ are replaced with $a_L$;
	\item for $\ell < m <r$ maximal blocks $a^m$ are replaced with $a_ma_L$;
	\item for $m \geq r$ maximal blocks $a^m$ are replaced with $a_Ra_ma_L$.
\end{itemize}
As in the case of the normal block compression,
for $m = 1$ we identify $a_1$ with $a$ (and do not make any replacement) and allow
further in the phase the compression of pairs including $a$.

The compression of blocks for $r<\ell$ is similar:
\begin{itemize}
	\item for $m<r$ blocks $a^m$ are replaced by $a_m$;
	\item for $r \leq m <\ell$ blocks $a^m$ are replaced with $a_Ra_m$;
	\item for $m = \ell$ blocks $a^\ell$ are replaced with $a_Ra_L$;
	\item for $m > \ell$ blocks $a^m$ are replaced with $a_Ra_ma_L$.
\end{itemize}
Again, we identify $a_1$ with $a$ (and do not make any replacement) and allow
further compression of pairs including $a$.

After the block compression, regardless of the actual scheme,
we compress pairs of the form $\makeset{a_Lb}{b \in \Sigma \setminus \{a_L\}}$ and then $\makeset{ba_R}{b \in \Sigma \setminus \{a_R\}}$.
Since in one such group pairs do not overlap, this can be easily done in linear time using \algradix,
as in the case of pair compression in \algSET.
(When compressing the second group of pairs we allow compression of letters introduced in the compression in the first group.)

Lastly, there is a special case: when $1 = r < \ell$ the compression of the pairs $\makeset{a_1b}{b \in \Sigma \setminus a_1}$
is also performed.
The running time is again linear.
When the \mytext{} after the block compression begins (ends) with $a_R$ ($a_L$, respectively), we remove it from \mytext,
as this letter cannot be used by any pattern appearance anyway.

Clearly both the beginning and end were fixed during the block compression,
we still need to guarantee that pattern appearances were not lost nor gained in the process.
The argument is similar as in Lemma~\ref{lem:compression of ends different}.
So let $\pattern = a^\ell w_2 a^r$, where $w \neq \epsilon$,
observe that we can make this assumption as we do not consider the case in which $\pattern \in a^*$.
Let $\mytext = w_1a^mw_2a^nw_3$, where $m \geq \ell$ and $n \geq r$.
There are several cases, we focus on one, the other are shown in the same way.
Suppose that $m > \ell >r$ and $\ell > n > r$.
Let $w_i$ be replaced by $w_i'$.
Then $\patterni = a_Lw_2'a_R$, while $\mytexti = w_1'a_Ra_ma_Lw_2'a_Ra_nw_3'$,
so there is an appearance of the pattern.
The other cases are shown similarly.

In the other direction, suppose that $\patterni = a_L w_2'a_R$ appears in $\mytexti = w_1'a_L w_2'a_Rw_3'$.
Observe that $w_2'$ in both was obtained from the same $w_2$, furthermore the only way to obtain $a_L$ ($a_R$) in $\mytexti$
is from $a^m$ ($a^n$, respectively) for some $m \geq \ell$ (some $n \geq a^r$, respectively).
Thus $\pattern = a^\ell w_2 a^r$ appeared in $\mytext = w_1 a^m w_2 a^n w_3$.

So not it is left to show that the following pair compressions do not spoil the beginning or end of the (new) pattern.
Consider the first compression of the pairs of the form $\makeset{a_Lb}{b \in \Sigma \setminus \{a_L\}}$:
is it possible that it spoils the end?
This can happen, when the last letter of the pattern, i.e.\ $a_R$ is compressed with a letter to its right.
Hence, $a_L = a_R$, which is not possible, as $a_L$ and $a_R$ are different symbols.
So consider the second compression phase, in which pairs of the form
$\makeset{ba_R}{b \in \Sigma \setminus a_R}$ are compressed.
Suppose that the beginning was spoiled in the process.
Let $b$ be the letter compressed with the leading $a_L$ in the pattern
(by the assumption that $\pattern \notin a^*$, such $b$ exists)
and let $c$ be the fresh letter that replaced $a_Lb$.
Then the beginning is spoiled, when pair of the form $xc$ is compressed,
but this implies $a_R = c$, which is not possible.
Lastly, consider the special case, i.e.\ $r=1 < \ell$,
in which additionally pairs of the form $a_1x$ were compressed.
This cannot spoil the end, as the last letter of \pattern{} is not $a_1$.
Suppose that this spoils the beginning.
We already know that $a_Lb$ was replaced with $c$.
As already shown, it could not be compressed with the letter to the left,
however, it is possible that it was compressed with the letter to the right, and replaced with $c'$.
Still the only possibility to spoil the beginning is to compress $a_1c$ or $a_1c'$, depending on the case.
In both cases this implies that before first compression phase there was a substring $a_Ra_1a_Lb$,
which contradicts our replacement scheme.
\qed
\end{proof}

Now, when the whole replacement scheme is defined, it is time to show that fixing 
preserves the main property of original \algSET:
that in each round the lengths of \pattern{} and \mytext{} are reduced by a constant factor.
Roughly, our replacement schema took care of that: for instance, even though we replaced a single $a$ with $a_Ra_L$,
we made sure that $a_R$ is merged with a previous letter and $a_L$ is merged with a following letter.
Effectively we replaced $3$ letters with $2$.
This is slightly weaker than replacing $2$ letters with $1$, but still shortens by a constant factor.
The other cases are analysed similarly.
The following lemma takes care of the details.
 
\begin{lemma}
\label{lem:fixing shortens}
When $|\pattern|, |\mytext| > 1$ then one phase of \algSPM{} shortens those lengths by a constant factor.
\end{lemma}
\begin{proof}
We group the compressed substrings into \emph{fragments},
one fragment shall intuitively correspond to small substring that was compressed into some letters.
Letters, that were not altered, are not assigned to fragments.
We show that there is a grouping of letters in \pattern{} and \mytext{} into fragments (in the beginning of the phase) 
such that 
\begin{enumerate}[({Fra} 1)]
	\item there are no consecutive letters not assigned to fragments; \label{partition 1}
	\item fragments of length $2$ are compressed into one letter till the end of the phase; \label{partition 2}
	\item fragments of length $3$ are compressed into at most two letters till the end of the phase; \label{partition 3}
	\item fragments of length $4$ or more are compressed into at most three letters till the end of the phase. \label{partition 4}
\end{enumerate}
This shows that the compression ratio is a little weaker then in case of \algSET,
see Lemma~\ref{lem:SET shortens}, but still by a constant factor.
So it is left to show that fixing, followed by block compression and pair compression,
allows grouping into fragments satisfying \Frarefall.

\begin{clm}
\label{lem:compression works different}
When the first and last letter of the pattern are different,
\Frarefall{} hold.
\end{clm}
\begin{proof}

Suppose that there are two consecutive letters not assigned to fragments, let them be $ab$.
They were not replaced in the phase, i.e.\ they are not fresh letters.
The analysis splits, depending on whether $a=b$ or not.
\begin{description}
	\item[$a=b$]
	Then this pair of consecutive letters is either compressed in the fixing of the beginning and end or it is going to be compressed
	in line~\ref{spm block compression}, contradiction.
	\item[$a \neq b$]
	Then the pair $ab$ is either compressed in the fixing of the beginning and end or it is going to be compressed
	in line~\ref{spm pair compression} contradiction.
\end{description}
It is left to define the fragments satisfying \Fraref{2}--\Fraref{4}.
In most we replace pairs or blocks with one letter only,
so this clearly satisfies \Fraref{2}--\Fraref{4}.
There is an exception: when \pattern{} begins (ends) with $a^\ell$ for $\ell > 1$ ($b^r$ for $r>1$, respectively),
then $a^m$ for $m>\ell$ is replaced with $a_ma_\ell$ ($b_rb_m$, respectively).
However, as $\ell > 1$ ($r>1$, respectively), this shows that $m>2$ and thus fragments replaced with $2$ letters are of length at least $3$,
which shows \Fraref{2}--\Fraref{4}
\qed
\end{proof}

When the first and last letter of the pattern are the same, the proof follows a similar idea.
We need to accommodate the special actions that were performed during the fixing of the beginning and end.

\begin{clm}
\label{lem:compression works same}
When the first and last letter of the pattern are the same \Frarefall{} hold.
\end{clm}
\begin{proof}

Let the first and last letter of \pattern{} be $a$.
Except for blocks of $a$ (and perhaps letters neighbouring them), all fragments are defined in the same way
as in Claim~\ref{lem:compression works different}, so we focus on the blocks of $a$.

The $a^m$ blocks for $m < \min(\ell,r)$ are replaced in the same way as in Claim~\ref{lem:compression works different},
so we deal mainly with $a^m$ for $m \geq \min(\ell,r)$. 
Let us first consider a simpler case, in which $\ell, r >1$.
Since the fragments depend also on the letters neighbouring the $a$-blocks,
take the longest possible substring of \pattern{} (or \mytext) of the form
$$
x^{(1)}a^{m_1}x^{(2)}a^{m_2}x^{(3)} \cdots x^{(k)}a^{m_k} x^{(k+1)},
$$
where $x^{(i)} \in \Sigma$ and $m_i > 1$.
Such substrings cover all blocks of $a$ except the leading and ending $a$-blocks in pattern and text.
To streamline the analysis, we deal with them separately at the end.

During the replacement, each block $a^{m_i}$ may introduce a letter $a_L$ to the right,
but it is compressed with $x^{(i+1)}$ and letter $a_R$ to the left, which is compressed with $x^{(i)}$.
Then the block is replaced with a single letter $a_{m_i}$ (or no letter at all, when $\ell = r = m_i$).
Hence the resulting string is
$$
y^{(1)}a_{m_1}y^{(2)}a_{m_2}y^{(3)} \cdots y^{(k)}a_{m_k} y^{(k+1)},
$$
where each $y^{(i)} \in \Sigma$ and each $a_{m_i}$ is either a letter or $\epsilon$.
Then define the first fragment as $x^{(1)}a^{m_1}x^{(2)}$, which is replaced with $y^{(1)}a_{m_1}y^{(2)}$,
and each consecutive fragment as $a^{m_i}x^{(i+1)}$, for $i>1$, which is replaced with $a_{m_i}y^{(i+1)}$.
Since $m_i>1$, such fragments satisfy \Frarefall.
Each other fragment is defined as in Claim~\ref{lem:compression works different},
i.e.\ letters compressed into a single symbol form a fragment.
The same argument as in Claim~\ref{lem:compression works different} shows
that \Frarefall{} holds for such defined grouping.

Now we consider the special cases omitted in the previous analysis,
i.e.: $1 = \ell = r$, $1 = \ell < r$ and $1 = r < \ell$.
In these case we consider similar maximal substrings
$$
x^{(1)}a^{m_1}x^{(2)}a^{m_2}x^{(3)} \cdots x^{(k)}a^{m_k} x^{(k+1)},
$$
of \pattern{} (or \mytext), but we allow $m_i = 1$.
Observe that as in the previous case, each $a$ block is covered by such susbtrings,
except for the leading and ending $a$ blocks of \pattern{} and \mytext.
To streamline the argument, we consider them at the end.

The fragments are defined in the similar way: the first one as
$x^{(1)}a^{m_1}x^{(2)}$ and $a^{m_i}x^{(i+1)}$ for $i>1$.
It remains to show that \Frarefall{} hold in this case as well.
Note that when $m_i>1$ the analysis is the same as previously, so we skip it and focus on the case of $m_i = 1$.
There are three cases, depending on the relation between $\ell$ and $r$:
\begin{description}
	\item[$\ell = r = 1$]
	Then $a$ is replaced with $a_Ra_L$ and $a_R$ is merged with $x^{(i)}$ while $a_L$ with $x^{(i+1)}$.
	Hence, for $i>2$ the fragment  $a^{m_i}x^{(i+1)}$ is replaced with $y^{(i+1)}$ alone,
	and for $i=1$ the $x^{(1)}a^{m_1}x^{(2)}$ is replaced with $y^{(1)}y^{(2)}$.
	So \Frarefall{} hold in this case.
	\item[$1 = \ell < r$]
	Then $a$ is replaced with $a_L$ which is then merged with $x^{(i+1)}$ and the rest of the analysis follows as in the first case.
	\item[$1 = r < \ell$]
	In this case $a$ is replaced with $a_Ra_1$, then $a_R$ is merged with $x^{(i)}$.
	Furthermore, in this special case, $a_1$ is also compressed, to $x^{(i+1)}$.
	Now, the rest of the analysis follows as in the first case.
\end{description}
The rest of the argument follows as in the proof of Claim~\ref{lem:compression works different},
and so it is omitted.

Concerning the leading and ending $a$-blocks observe that in case of the \pattern,
the $a^\ell$ ($a^r$) is replaced with $a_L$ ($a_R$, respectively),
which is later compressed with the letter to the right (left, respectively).
So the leading $a^\ell$ (ending $a^r$) can be added to the fragment to its right (left, respectively)
and \Frarefall{} still holds.

For the leading $a$-block of \mytext, we extend the definition and consider a substring
$$
x^{(0)}a^{m_0}x^{(1)}a^{m_1}x^{(2)} \cdots x^{(k)}a^{m_k} x^{(k+1)},
$$
where $a^{m_0}$ is the leading $a$ block of \mytext{} and $x^{(0)} = \epsilon$ is an imaginary beginning marker.
Then the whole analysis works in the same way: the only difference is that $a_R$ that may be produced by $a^{m_0}$
to the left is removed from \mytext, which is simulated by 'merging' it into the imaginary beginning marker $x^{(0)}$.
Otherwise, the fragments are defined in the same way.
The analysis for the ending block of $a$s is similar.
\qed
\end{proof}
\qed
\end{proof}

Concerning other operations, they are implemented in the same way as in case of \algSET,
so in particular the pair compression and block compression run in $\Ocomp(|\pattern| + |\mytext|)$,
see Lemma~\ref{lem:SET linear}.
Furthermore, since the beginning and end are fixed,
those operations do not spoil pattern appearances, see Lemma~\ref{lem:when appearances are preserved}.

As a corollary we are now able to show that \algSPM{} runs in linear time and preserves the appearances of the pattern,
which follws from Lemma~\ref{lem:compression of ends different} and \ref{lem:compression of ends same}.
\begin{lemma}
\algSPM{} runs in $\Ocomp(n+m)$ time and correctly reports all appearances of a pattern in text.
\end{lemma}
The running time is clear: each phase takes linear time and the length of text and pattern are shortened by a constant factor in a phase.

\subsubsection*{Building of a grammar revisited}
Note that the more sophisticated replacement rules in the fixing of beginning and end endangers our
view of compression as creation of a context free grammar for \pattern{} and \mytext.
Still, this can be easily fixed.

For the fixing of the beginning when the first and last letter are different
there are symmetric actions performed at the beginning and at the end, so we focus only on the former.
The problematic part is the replacement of $a^m$ for $m > \ell$ with $a_ma_\ell$.
Then we simply declare that $a_l$ replaced $a^\ell$ (note that this is consistent with the fact that $a^\ell$ is replaced with $a_\ell$)
and $a_m$ replaced $a^{m-\ell}$. Since $m > \ell$, this is well defined.

When the first and last letter are the same, the situation is a bit more complicated.
For the block replacement, similarly we declare that $a_L$ `replaces' $a^\ell$,
$a_m$ the $a^m$ for $m < \ell$ and $a^{m-\ell}$ for $m >\ell$.
Lastly, to be consistent, we need to define that $a_R \to \epsilon$.
It can be verified by case inspection that in this way all blocks are replaced properly, except the ending block for \pattern{}
(for which the $a^r$ is replaced with $a_R$).
While this somehow falsifies our informal claim that we create an SLP for the \pattern,
this is not a problem, as the occurrences of the pattern are preserved.
(We can think that we shortened the pattern by those ending $a^r$ letters, but the appearances were preserved.)

The $a_R$ generating $\epsilon$ is a bit disturbing,
but note that we enforce the compression of pairs of the  form $\makeset{ba_R}{b \in \Sigma \setminus a_R}$
(and if $a_R$ is the first letter of \mytext{} then we remove it).
In this way all $a_R$ are removed from the instance.
Furthermore, when $ba_R$ is replaced with $b'$ we can declare that the rule for $b'$ is $b' \to \alpha$,
where $b$ has a rule $b \to \alpha$.
In this way no productions have $\epsilon$ at their right-hand sides.

\section{Equality testing for SLPs}
\label{sec:outline}
In this section we extend the \algSET{} to the setting in which both the \pattern{} and \mytext{}
are given using SLPs.
In particular, we introduce and describe the second important property of the recompression:
local modifications of the instance so that pair and block compressions
can be performed on the compressed representation directly.

\subsection{Straight line programmes}
Formally, a \emph{Straight-Line Programme} (SLP) is a context free grammar $G$
over the alphabet $\letters$ with a set of nonterminals $\{X_1, \ldots, X_k\}$,
generating a one-word language.
For normalisation reasons, it is assumed that $G$ is in a \emph{Chomsky normal form},
i.e.\ each production is either of the form $X \to YZ$ or $X \to a$. 
We denote the string defined by nonterminal $X$ by $\eval(X)$, like \emph{value}.

During our algorithm, the alphabet $\letters$ is increased many times
and whenever this happens, the new letter is assigned number $|\letters| +1$.
The $|\letters|$ does not become large in this way: it remains of size
$\Ocomp((n+m) \log (n+m) \log M)$, see Lemma~\ref{lem:running time}.
Observe furthermore that Claim~\ref{clm:consecutive values} generalises easily to SLPs
and so without loss of generality we may assume that $\letters$ consists of consecutive
natural numbers (starting from $1$).

For our purposes it is more convenient to treat
the two SLPs as a single context free grammar $G$
with a set of nonterminals $\{ \allXc \}$,
the text being given by \textc{} and the pattern by \patternc.
We assume, however, that \patternc{} is not referenced by any other nonterminal,
this simplifies the analysis.
Furthermore, in our constructions, it is essential to \emph{relax} the usual assumption
that $G$ is in a Chomsky normal form,
instead we only require that $G$ satisfies the conditions:
\begin{subequations}
\label{form}
\begin{align}
	\label{one production}
	&\text{each $X_i$ has exactly one production, which has at most $2$ noterminals},\\
	&\text{if $X_j$ appears in the rule for $X_i$ then $j<i$, \label{eq:rule form}}\\
\label{trivial epsilon rules}
&\text{if } \eval(X_i) = \epsilon \text{ then } X_i \text{ is not 
	on the right-hand side of any production,}.
\end{align}
\end{subequations}
We refer to these conditions collectively as~\eqref{form} and 
assume that the input of the subroutines always satisfies \SLPrefall.
However, we expect more from the input instance: we want it to obey the Chomsky normal form,
instead of the relaxed conditions~\eqref{form}
(in this way we bound the initial size of $G$ by $2(n+m)$ and also claim that $M \leq 2 ^ m$).
Note that~\eqref{form} does not exclude the case,
when $X_i \to \epsilon$ and allowing such a possibility streamlines the analysis.

Let $X_i \to \alpha_i$, then a substring $u \in \letters^+$ of $\alpha_i$
appears \emph{explicitly} in the rule;
this notion is introduced to distinguish them from the substrings of $\eval(X_i)$.
The size $|G|$ is the sum of length of the right-hand sides of $G$'s rules.
The size of $G$ kept by the algorithm will be small: $\Ocomp((n+m)\log(n+m))$,
see Lemma~\ref{lem:running time}.
Furthermore the set of nonterminals is always a subset of $\{\allXc\}$.

\subsubsection*{(Non) crossing appearances}
The outline of the algorithm is the same as \algSET, the crucial difference is the way we want
perform the compression of pairs and blocks, when \pattern{} and \mytext{} are given as SLPs.
Before we investigate this, we need to understand, when the compression (of pairs and blocks) is easy to perform,
and when it is hard.

Suppose that we are to compress a pair $ab$.
If $b$ is a first letter of some $\eval(X_i)$ and $aX_i$ appears explicitly
in the grammar, then the compression seems hard, as it requires modification of $G$.
On the other hand, if none such, nor symmetrical, situation appears
then replacing all explicit $ab$s in $G$ should do the job.
This is formalised in the following definition:

\begin{definition}[(Non) crossing pairs]
Consider a pair $ab$ and its fixed appearance in $\eval(X_i)$, where the rule for $X_i$ is $X_i \to u X_j v X_k w$.
We say that this appearance is
\begin{description}
	\item[explicit (for $X_i$)] if this $ab$ comes from $u$, $v$ or $w$;
	\item[implicit (for $X_i$)] if this appearance comes from $\eval(X_j)$ or $\eval(X_k)$;
	\item[crossing (for $X_i$)] otherwise.
\end{description}

A pair $ab$ is \emph{crossing} if it has a crossing appearance for any $X_i$,
it is \emph{non-crossing} otherwise.
\end{definition}
Unless explicitly written, we use this notion only to pairs of \emph{different} letters.
Note that if $ab$ appears implicitly in some $X_i$ then it has an explicit or crossing appearance in some $X_j$
for $j < i$.

The notions of (non-) crossing pairs is usually not applied to pairs of the form $aa$,
instead, for a letter $a \in \letters$ we consider its maximal blocks,
as defined in earlier sections.
\begin{definition}
Consider a letter $a$.
We say that $a^\ell$ has an explicit appearance in $X_i$ with a rule $X_i \to u X_j v X_k w$
if $a^\ell$ appears in $u$, $v$ or $w$; implicit appearance if it appears in $\eval(X_j)$ or $\eval(X_k)$
and a crossing appearance if it appears in $\eval(X_i)$ and this is not an implicit, nor explicit appearance.

A letter $a$ has a crossing block, if some $a^\ell$ has a crossing appearance in some $X_i$.
Equivalently, the pair $aa$ is crossing.
\end{definition}

Note that when $a$ has crossing blocks it might be that some blocks of $a$ are part of explicit and crossing appearances at the same time.
However, when $a$ has no crossing blocks, then a maximal explicit block of $a$ is not part of a larger crossing block.

Intuitively, a pair $ab$ is crossing, if we can find a rule $X_i \to u X_j v X_k w$
such that $a$ is the last letter of $u$ and $b$ is the first letter or $\eval(X_j)$,
or $a$ is the last letter of $\eval(X_j)$ and $b$ is the first letter of $v$, \textit{etc}.
So in some sense it `crosses' between this nonterminal and a neighbouring letter (nonterminal).
Note that this justifies the somehow unexpected notion of crossing blocks:
if $aa$ is crossing pair (say, $\eval(X_i)$ ends with $a$ and $v$ begins with $a$ as well)
then the maximal block of $a$s containing this pair $aa$ also `crosses' the nonterminal.

The crossing pairs and letters with crossing blocks are intuitively hard to compress,
while non-crossing pairs and letters without crossing blocks are easy to compress.
The good news is that the number of crossing pairs and blocks is bounded in terms of $n,m$, and
not size of the grammar, as shown in the following lemma.
Note that the lemma allows a slightly more general form of the grammar, in which
blocks $a^\ell$ are represented using a single symbol.
Such a form appears as an intermediate product of our algorithm, and so we need to deal with it as well.

\begin{lemma}
\label{lem:different crossing}
Consider a grammar, in which blocks of a letter can be represented as a single symbol.
There are at most $2(n+m)$ different letters with crossing blocks
and at most $4(n+m)$ different crossing-pairs and at most $|G|$ noncrossing pairs.
For a letter $a$ there are at most $|G| + 4(n + m)$ different lengths of
$a$'s maximal blocks in \pattern{} and \mytext.
\end{lemma}
\begin{proof}
Observe that if $a$ has a crossing block then for some $X_i$ the first or last letter of $\eval(X_i)$ is $a$.
Since there are $n + m$ nonterminals, there are at most $2(m+n)$ letters with crossing blocks.

Similarly, if $ab$ is a crossing pair then it can be associated with an appearance of some $X_i$ in the grammar,
where additionally $a$ is the last letter of $\eval(X_i)$ and $X_ib$ appears in the rule
or $b$ is the first letter of $\eval(X_i)$ and $aX_i$ appears in the rule.
Since there are at most $2(n+m)$ appearances of nonterminals in the grammar,
it follows that there are at most $4(n+m)$ appearances of a crossing pair,
so in particular at most $4(n+m)$ different crossing pairs.

If $ab$ is a noncrossing pair then $ab$ appears explicitly is some of the rules of the grammar,
and there are at most $|G|$ such substrings
(note that when $a^\ell$ is represented by one symbol, it still contributes to pairs in the same way as a single $a$).

The argument for maximal blocks of $a$ is a little more involved.
Consider first maximal blocks that have an explicit appearance in the rules of $G$,
for simplicity let now the nonterminals also count for ending maximal blocks,
similarly the ends of rules.
Then each letter (or block of letters that are represented as one symbol) is
assigned to at most one maximal block and so there are not more than $|G|$ such blocks,
so not more than $|G|$ different lengths.
Assign other blocks to nonterminals: a block $a^\ell$
is assigned to $X_i$ with a rule $X_i \to uX_jvX_kw$,
if a maximal block $a^\ell$ has an appearance in $\eval(X_i)$,
but it does not in $\eval(X_j)$ nor in $\eval(X_k)$ (so it has a crossing appearance for $X_i$).
Thus, there are four possibilities for a block to be assigned to the rule:
\begin{itemize}
	\item a letter $a$ from this maximal block is the last letter of $u$
	and the first letter of $\eval(X_j)$,
	\item $a$ is the last letter of $\eval(X_j)$ and a letter $a$ from this maximal block is the first letter of $v$,
	\item a letter $a$ from this maximal block is the last letter of $v$ and $a$ is the 
	first letter of $\eval(X_k)$,
	\item $a$ is the last letter of $\eval(X_k)$ and a letter $a$ from this maximal block is the first letter of $w$.
\end{itemize}
Hence, there are at most $4$ maximal blocks assigned to $X_i$ in this way, which yields the desired bound of $4(n+m)$
on the number of such blocks.
\qed
\end{proof}

\subsection{The algorithm}
When the notions of crossing and non-crossing pairs (blocks) are known,
we can give some more detail of \algET.
Similarly to \algSET, it performs the compression in phases,
until one of \pattern, \mytext{} has only one letter,
but for running time reasons it is important to distinguish between
compression of non-crossing pairs and crossing ones
(this is not so essential for blocks, as shown later).
\begin{algorithm}[H]
	\caption{\algET: outline}
	\begin{algorithmic}[1]
	\While{$|\pattern|,|\mytext|>1$} \label{alg:mainloop}
		\State $P \gets $ list of pairs 
		\State \label{listing non-crossing} $L \gets $ list of letters
    \For{each $a \in L$} {} compress blocks of $a$ \label{outer appearances}
		\EndFor \label{end of block compression}
		\State $P' \gets $ crossing pairs out of $P$, $P \gets $ non-crossing pairs out of $P$
    \For{each $ab\in P$} {} compress pair $ab$ \label{noncrossing compression}
  	\EndFor \label{end of non-crossing compression}
		\For{$ab \in P'$} \label{ac list} compress pair $ab$ \label{crossing compression} \label{compress a}
		\EndFor \label{end crossing compression}
		\EndWhile
	\State Output the answer.
	 \end{algorithmic}
\end{algorithm}

As in the case of \algSET, the length of \pattern{} and \mytext{}
shorten by a constant factor in a phase and so there are $\Ocomp(\log (\min(M,N)))$ many phases.
\begin{lemma}
\label{lem:logM iterations SE}
There are $\Ocomp(\log M )$ executions of the main loop of \algFCPM.
\end{lemma}
The proof is the same as in the case of Lemma~\ref{lem:SET shortens}.

\subsubsection{Compression of non-crossing pairs}
\label{subsec: nc pair compression}
We start by describing the compression of a non-crossing pair $ab$,
as it is the easiest to explain.
Intuitively, 
whenever $ab$ appears in string encoded by $G$,
the letters $a$ and $b$ cannot be split between nonterminals.
Thus, it should be enough to replace their explicit appearances.

\begin{algorithm}[H]
	\caption{$\algpairsncr(ab,c)$: compression of a non-crossing pair $ab$}
	\label{alg:noncrossing pairs}
	\begin{algorithmic}[1]
		\For{$i \gets 1 \twodots m+n$}
			\State replace every explicit $ab$ in the rule for $X_i$ by $c$
		\EndFor
	 \end{algorithmic}
\end{algorithm}

Luckily, as in case of \algSET{} the compression of all noncrossing pairs can be performed in parallel in linear time,
using \algradix{} to group the appearances.

To simplify the notation, we use $\PC_{ab \to c}(w)$ to denote $w$ with each $ab$ replaced by $c$.
Moreover, we say that a procedure \emph{implements the pair compression} for $ab$,
if after its application the obtained $\patterni$ and $\mytexti$ satisfy
$\patterni = \PC_{ab \to c}(\pattern)$ and $\mytexti = \PC_{ab \to c}(\mytext)$.

\begin{lemma}
When $ab$ is non-crossing, \algpairsncr{} properly implements the pair compression.
\end{lemma}
\begin{proof}
In order to distinguish between the nonterminals before and after
the compression of $ab$
we use `primed' nonterminals, i.e.\  $X_i'$,
for the nonterminals after this compression and `unprimed',
i.e.\ $X_i$, for the ones before.
We show by induction on $i$ that
\begin{equation*}
\eval(X_i') = \PC_{ab \to c}(\eval(X_i)) \enspace .
\end{equation*}
Indeed, this is true when the production for $X_i$ has no nonterminal on the right-hand side
(recall the assumption that $a \neq b$),
as in this case each pair $ab$ on right hand side of the production for $X_i$ was replaced by $c$ and so
$\eval(X_i') = \PC_{ab \to C}(\eval(X_i))$.

When $X_ i \to u X_j v X_kw$, then
\begin{align*}
\eval(X_i)
	&=
u\eval(X_j)v\eval(X_k)w \quad \text{and}\\
\eval(X_i')
	&=
\PC_{ab \to c}(u) \eval(X_j') \PC_{ab \to c}(v) \eval(X_k') \PC_{ab \to c}(w)\\
	&=
\PC_{ab \to c}(u) \PC_{ab \to c}(\eval(X_j)) \PC_{ab \to c}(v) \PC_{ab \to c}(\eval(X_k)),
\end{align*}
with the last equality following by the induction assumption.
Notice, that since $ab$ is a non-crossing pair, all occurrences of $ab$ in $\eval(X_i)$
are contained in $u$, $v$, $w$, $\eval(X_j)$ or $\eval(X_k)$,
as otherwise $ab$ is a crossing pair, which contradicts the assumption.
Thus,
\begin{align*}
\PC_{ab \to c}(\eval(X_i)) = 
\PC_{ab \to c}(u) \PC_{ab \to c}(\eval(X_j')) \PC_{ab \to c}(v) \PC_{ab \to c}(\eval(X_k'))\PC_{ab \to c}(w),
\end{align*}
which shows that $\PC_{ab \to c}(\eval(X_i)) = \eval(X_i')$.
\qed
\end{proof}

As in the case of \algSET, the pair compression of all noncrossing pairs can be effectively implemented, with a help of \algradix{}
for grouping of the appearances.

\begin{lemma}
\label{lem:noncrossing and inner}
The non-crossing pairs compression can be performed in $\Ocomp(|G|)$ time.
\end{lemma}
\begin{proof}
We go through the list productions of $G$.
Whenever we spot an explicit pair $ab$, we put
$(a,b,1,p)$ in the list of pairs' appearances, where $1$ indicates, that
this appearance is non-crossing and $p$ is the pointer to the appearance in $G$.

It is easy to list the crossing pairs:
we begin with calculating for each nonterminal $X_i$ the first and last letter of $\eval(X_i)$,
which can be easily done in a bottom-up fashion.
Then for $aX_ib$ appearing in the right-hand side of a rule we list the tuples for pairs $af$ and $\ell b$
with flag $0$ indicating, that they are crossing, where
$f$ ($\ell$) is the first (last, respectively) letter in $\eval(X_i)$
(the pointer $p$ is not important, as it is not going to be used for anything)

Then, we sort all these tuples lexicographically, using \algradix{} in $\Ocomp(|G|)$ time:
by Lemma~\ref{lem:running time} the size of $\letters$ is polynomial in $n+m$,
and \algradix{} sorts the tuples in $\Ocomp(|G| + n + m) = \Ocomp(|G|)$ time.
Thus, for each pair we obtain a list of its appearances.
Moreover, when sorted, we can establish in $\Ocomp(|G|)$ time,
which pairs are crossing and which non-crossing:
since $0 < 1$ the first appearance of $ab$ on the list will have $0$ on the third coordinate of the tuple if and only if the pair $ab$ is crossing.

For a fixed non-crossing pair $ab$, the compression is performed as in the case of \algSET,
see Lemma~\ref{lem:SET linear}:
We go through the associated list and use pointers to localise and replace
all appearances of $ab$.
If this pair is no longer there (as one of letters $ab$ was already replaced),
we do not nothing. For a crossing pair, we do nothing.

The correctness follows in the same way as in Lemma~\ref{lem:SET linear}, it only remains to estimate the running time.
Since rules of $G$ are organised as lists, the pointers can be manipulated in constant time,
and so the whole procedure takes $\Ocomp(|G|)$ time.
\qed
\end{proof}

\subsubsection*{Compression of crossing pairs}
We intend to reduce the case of crossing pairs to the case of non-crossing one, i.e.\ given a crossing pair we want to `uncross' it
and then compress using the procedure for compression of noncrossing pairs, i.e.\ \algpairsncr.

Let $ab$ be a crossing pair. Suppose that this is because $a$ is to the left of nonterminal $X_i$
such that $\eval(X_i) = bw$.
To remedy this we `left-pop' the leading $b$ from $X_i$: we modify $G$ so that $\eval(X_i) = w$ and replace each  $X_i$ with $bX_i$ in the rules.
We apply this procedure to each nonterminal, in an increasing order.
It turns out that the condition that $X_i$ is to the right of $a$ is not needed, we left-pop $b$ whenever $X_i$ starts with it.
Symmetric procedure is applied for a letter $a$ and nonterminals $X_i$ such that $\eval(X_i)=w'a$.
It can be easily shown that after left-popping $b$ and right-popping $a$ the pair $ab$ is no longer crossing,
and so it can be compressed.

Uncrossing a pair $ab$ works for a fixed pair $ab$ and so it has to be applied to each crossing pair separately.
It would be good to uncross several pairs at the same time.
In general it seems impossible to uncross an arbitrary set
of pairs at the same time.
Still, parallel uncrossing can be done for group of pairs of a specific form: when we partition the alphabet \letters{}
to $\letters_\ell$ and $\letters_r$ then pairs from $\letters_\ell\letters_r$ can be uncrossed in parallel.
Intuitively, this is because pairs from $\letters_\ell\letters_r$ cannot overlap
as the same letter cannot be at the same time the first in some crossing pair in this group and a second one.
Furthermore, using a general construction (based on binary expansion of numbers),
we can find $\Ocomp(\log(n+m))$ partitions
such that each of $4(n+m)$ crossing pairs is covered by at least one of those partitions.

Note that letters should not be popped from \patternc{} and \textc:
on one hand those nonterminals are not used in the rules and so they cannot be used to create a crossing pair,
on the other hand, since they define \pattern{} and \mytext{} we should not apply popping to them,
as this would change text or pattern.

\begin{algorithm}[H]
  \caption{$\algpop(\letters_\ell,\letters_r)$: Popping letters from $\Sigma_\ell$ and $\Sigma_r$
  \label{pop code full}}
  \begin{algorithmic}[1]
  \For{$i \gets 1 \twodots n+m$, except $m$ and $m+n$}
		\State let $X_i \to \alpha_i$ and $b$ the first letter of $\alpha_i$
		\If{the first letter $b \in \Sigma_r$} \label{still to the right} \Comment{Left-popping}
			\State remove leading $b$ from $\alpha_i$
			\State replace $X_i$ in $G$'s rules by $bX_i$
			\If{$\alpha_i = \epsilon$}	{} remove $X_i$ from rules of $G$ \Comment{$X_i$ is empty} 
			\EndIf
		\EndIf
		\State let $a$ be the last letter of $\alpha_i$
		\If{$a \in \Sigma_\ell$} \label{still to the left} \Comment{Right-popping}
			\State remove ending $a$ from $\alpha_i$
			\State replace $X_i$ in $G$'s rules by $X_ia$
			\If{$\alpha_i = \epsilon$}	{} remove $X_i$ from rules of $G$ \Comment{$X_i$ is empty} 
			\EndIf
		\EndIf
	\EndFor
  \end{algorithmic}
\end{algorithm}

\begin{lemma}
\label{lem:uncrossing pairs}
After $\algpop(\letters_\ell,\letters_r)$ no pair in $\letters_\ell\letters_r$ is crossing.
Furthermore, $\eval(\patternc)$ and $\eval(\textc)$ have not changed.

\algpop{} runs in time $\Ocomp(n+m)$ and introduces at most $4(n+m)$ letters to $G$.
\end{lemma}
\begin{proof}
Suppose that $ab \in \letters_\ell,\letters_r$ is crossing after $\algpop(\letters_\ell,\letters_r)$.
Without loss of generality consider the case, in which after $\algpop(\letters_\ell,\letters_r)$
there is $aX_j$ in the rule for $X_i$ and $\eval(X_j)$ starts with $b$.
We first show by induction, that if $\eval(X_j)$ started with a letter from $\letters_r$ then this letter
was left-poppped from $X_j$ by \algpop.
This is of course true for $X_1$, for general $X_j$ with a rule $X_j \to \alpha_j$
consider that if $\eval(X_j)$ begins with $b \in \letters_r$, in which case it is left-popped,
or with $X_{k}$, where $k < j$.
In the latter case \algpop{} did not pop a letter from $X_k$.
As $\eval(X_k)$ begins with $b \in \letters_r$ it should have, contradiction.

Returning to the main claim, we want to show that it is impossible that after \algpop{}
the $aX_j$ appears in the rule for $X_i$, where $\eval(X_j)$ begins with $b \in \letters_r$.
Consider, whether $\algpop(\letters_\ell,\letters_r)$ left-popped a letter from $X_j$.
If so, then it replaced $X_j$ with $cX_j$ and letter $c \in \letters_r$
cannot be changed to any other letter during the whole $\algpop(\letters_\ell,\letters_r)$.
Hence $a = c \in \letters_r$, which is a contradiction.
If no letter was popped from $X_j$,
then its first letter is not changed afterwards, and so it is $b \in \letters_r$.
However, $b$ should have been popped from $X_j$, contradiction.

The other cases are shown in the same way.

Concerning the running time note that we do not need to read the whole $G$: it is enough to read the first and last letter in each rule.
To perform the replacement, for each nonterminal $X_i$ we keep a list of pointers to its appearances, so that $X_i$ can be replaced with
$aX_ib$ in $\Ocomp(1)$ time.

Note that at most $2$ letters are popped from each nonterminal and so at most $4(n+m)$ are introduced to $G$.
\qed
\end{proof}

Now, when the pairs $ab \in \letters_\ell\letters_r$ are no longer crossing, we can compress them.
Since such pairs do not overlap, this can be done in parallel in linear time,
similarly as in Lemma~\ref{lem:noncrossing and inner}.

The obvious way to compress all crossing pairs, is to make a series of partition
$(\letters_\ell^{(1)},\letters_r^{(1)})$, $(\letters_\ell^{(2)},\letters_r^{(2)})$, \ldots 
such that each crossing pair is in at least one of those partitions.
Since there are $4(n+m)$ crossing pairs, see Lemma~\ref{lem:different crossing},
in the naive solution we would have $4(n+m)$ partitions.
However, using a more clever approach, we can reduce this number to $\Ocomp(\log(n+m))$.
Roughly, we make the partitions according to the binary expansion of notations of letters:
For $i = 1, \ldots \lceil \log|\Sigma| \rceil$, define
\begin{itemize}
	\item $\Sigma_\ell^{(2i-1)} = \Sigma_r^{(2i)}$ consist of elements of $\Sigma$ that have $0$ at the $i$-th
	position in the binary notation (counting from the least significant digit)
	\item $\Sigma_r^{(2i-1)} = \Sigma_\ell^{(2i)}$ consist of elements of $\Sigma$ that have $1$ at the $i$-th
	position in the binary notation
\end{itemize}
For $a \neq b$, their binary notation differ at some position and so
the pair $ab$ is in some group $\Sigma_\ell^{(j)}\Sigma_r^{(j)}$.
Note, that $ab$ may be in many $\Sigma_\ell^{(j)}\Sigma_r^{(j)}$  but it will be compressed only once,
for the smallest possible $j$. Thus, it makes sense to define the lists $P_j$,
where we include $ab \in P'$ in the group $P_j$, when $j$ is the smallest number such that
$a \in \Sigma_\ell^{(j)}$ and $b \in \Sigma_r^{(j)}$.
Observe, that using standard bit operations we can calculate the first position on which $a$ and $b$
differ and so also $j$ for $ab$ in constant time.
Lastly, since $|\Sigma| = \Ocomp((n+m)\log(n+m)\log M) = \Ocomp((n+m)^3)$ by Lemma~\ref{lem:running time},
we partition $P'$ into at most $\Ocomp(\log(n+m))$ subgroups.

\begin{algorithm}[H]
  \caption{$\algpairscr(\letters_\ell,\letters_r)$ compressing crossing pairs from $\letters_\ell\letters_r$.
  \label{alg:crossing compression}}
  \begin{algorithmic}[1]
		\State find partitions of $\Sigma$ into $\{\Sigma_\ell^{(i)},\Sigma_r^{(i)}\}$, $i \in \Ocomp(\log(n + m))$ \Comment{see discussion above}
		\State partition the crossing pairs into groups $P_1$, $P_2$, \ldots, $P_{2i}$ according to partitions of $\Sigma$
		\For{$j\gets 1 \twodots 2i$} \label{compressing crossing}
			\State run $\algpop(\Sigma_\ell^{(j)},\Sigma_r^{(j)})$ \label{uncrossing a group}
			\State compress each of the pairs $ab \in P_j$ \label{compressing no longer crossing}
					\Comment{$P_j$ is more or less $P' \cap \Sigma_\ell^{(j)}\Sigma_r^{(j)}$}
		\EndFor
	\end{algorithmic}
\end{algorithm}

Concerning the running time of an efficient implementation,
we first compute the list of explicit appearances of each crossing pair, 
which is done in linear time using the same methods as in the case of noncrossing pairs
and divide those pairs into groups, also in linear time.
However, \algpop{} creates new explicit appearances of pairs, which should be also compressed.
Still, we can easily identify those new appearances and assign them to appropriate groups.
Re-sorting each group before the compression makes sure that we can replace the pairs.

\begin{lemma}
\label{lem:compression of all pairs}
The \algpairs{} properly compresses all crossing pairs.
It runs in $\Ocomp(|G| + (n + m) \log (n + m))$ time.
It introduces $\Ocomp(\log (n + m))$ letters to each rule.
\end{lemma}
\begin{proof}
The sorted list of all appearances of each crossing pair is obtained as a by-product of
creation a similar list for noncrossing pairs, see Lemma~\ref{lem:noncrossing and inner}.
Each pair $ab$ is assigned to the appropriate group $P_j$ (according to the partition for $\Sigma_\ell^{(j)},\Sigma_r^{(j)}$) in constant time.

We analyse the processing of a single group $P_j$.
By induction on the number of the group ($j$) we show the following claim:

\begin{clm}
\label{clm:compression group}
Compression of pairs from one group $P_j$, i.e.\ lines~\ref{uncrossing a group}--\ref{compressing no longer crossing},
can be done in time $\Ocomp(|P_j| + n + m)$.
\end{clm}
\begin{proof}
Firstly, by Lemma~\ref{lem:uncrossing pairs}, the application of $\algpop(\Sigma_\ell^{(j)},\Sigma_r^{(j)})$
takes time $\Ocomp(n + m)$ and afterwards the pairs from $P_j$ are non-crossing.
Note, that $\algpop(\Sigma_\ell^{(j)},\Sigma_r^{(j)})$ introduces new explicit pairs
to $G$: when we replace $X_i$ by $bX_i$ and $a$ is a letter to the left of $X_j$,
a new explicit pair $ab$ appears. In constant time we can decide, to which $P_{j'}$ this pair
should belong, we simply add it an appropriate tuple $(a,b,0,p)$ to the list $P_{j'}$ (which makes the list $P_{j'}$ unsorted).
There two remarks:
firstly, by inductive assumption all appearances of pairs from $P_{j''}$ for $j'' < j$
were already replaced and so $j' \geq j$, so the newly introduced pairs will be handled later;
secondly, we do not know in advance, whether the pair $ab$ is one of the crossing pairs and so whether it should be compressed.
To remedy this, each element $P_j$ stores also an information,
whether it was a crossing pair or perhaps it was added later on; those are used to decide whether $ab$ should be compressed
at all, as described later on.

Now, since we cannot assume that the records in the list $P_j$ are sorted or even that they shold be compressed at all
(as we might have added some pairs to $P_j$ when considering $P_{j'}$ for $j' < j$),
we sort them again, using \algradix, ignoring the coordinate for the pointers;
furthermore we add another coordinate, which is $1$ for original crossing pairs and $0$ for those introduced due to recompression.
The compression can be done in time $\Ocomp(|P_j| + n + m)$.
Now, as the list of appearances of pairs are sorted, we can cut it into several lists,
each consisting of appearances of a fixed pair.
Going through one list, say for a pair $ab$,
we first whether the first appearance is an original crossing pair,
if not, then we do not compress appearances of this pair at all.
If it is an original crossing pair, we replace appearances of $ab$ (if they are still there)
in $\Ocomp(|P_j|)$ time: since we replace appearances of one fixed pair, we replace always by the same (fresh) letter
and so do not need to use any dictionary operations to look-up the appropriate letter.
Clearly, this procedure properly implements the pair compression for a single pair $ab$
and thus also for all pairs in $P_j$ (note that pairs in $P_j$ cannot overlap).
\qed
\end{proof}

The running time of the whole loop~\ref{compressing crossing} is at most (for some constant $c$):
\begin{align*}
\sum_{j=1}^{2i} c(|P_j| + n + m)
	&=
2c(n + m)i + c\sum_{j=1}^{2i} |P_j|\\
	&= 
\Ocomp((n + m)\log(n + m)) + c\sum_{j=1}^{2i} |P_j| \enspace .
\end{align*}
It is tempting to say that $\sum_{j=1}^{2i} |P_j| \leq |G| + 4(n+m)$:
observe that before the loop~\ref{compressing crossing} each element of $\sum_{j=1}^{2i} |P_j|$ 
corresponds to some appearance of a (crossing) pair in $G$, and there are only $|G| + 4(n+m)$ such appearances
by Lemma~\ref{lem:different crossing}.
However, \algpop{} introduce new pairs to the lists. Still, there are only $2(n + m)$ pairs added
by one run of \algpop, see Lemma~\ref{lem:uncrossing pairs},
hence in total there are only $2i(n + m)$ pairs introduced in this way.
Hence
\begin{align*}
\sum_{j=1}^{2i} |P_j|
	&\leq
|G| + 4(n+m) + 2i(n+m)\\
	&=
\Ocomp(|G| + (n+m) \log (n+m)) \enspace .
\end{align*}
Thus, the total running time is $\Ocomp((n + m)\log(n + m) + |G|)$,
and at most $\Ocomp(\log(n + m))$ pairs are introduced into a rule.
\qed
\end{proof}

\subsection{Blocks compression}
Now, we turn our attention to the block compression.
Suppose first that $G$ has no letters with a crossing block.
Then a procedure similar to the one compressing non-crossing pairs can be performed:
when reading $G$, we establish all maximal blocks of letters.
We group these appearances according to the letter, i.e.
for each letter $a$ we create a list of $a$'s maximal blocks in $G$
and we sort this list according to the lengths of the blocks.
We go through such list and we replace each appearance
of $a^\ell$ by a fresh letter $a_\ell$.

However, usually there are letters with crossing blocks.
We deal with this similarly as in the case of crossing pairs:
a letter $a$ has a crossing block if and only if $aa$ is a crossing pair.
So suppose that $a$ is to the left of $X_i$ and the first letter of $\eval(X_i)$ is $a$,
in such a case we left-pop a letter from $X_i$.
In general, this does not solve the problem as it may happen that still $a$ is the first letter of $\eval(X_i)$.
So we keep on left-popping until it is not.
In other words, we remove the $a$-prefix of $\eval(X_i)$.
Symmetric procedure is applied to $X_j$ such that $a$ is the last letter of $\eval(X_j)$ and $X_j$ is to the left of $a$.

It turns out that even a simplified approach works:
for each nonterminal $X_i$, where the fist letter of $\eval(X_i)$ is $a$
and the last letter of $\eval(X_i)$ is $b$,
it is enough to `pop' its $a$-prefix and $b$-suffix,
see~\algremblocks.

Observe that during the procedure, long blocks of $a$ (up to $2^{n+m})$
may be explicitly written in the rules.
This is conveniently represented: $a^\ell$ is simply
denoted as $(a,\ell)$, with $\ell$ encoded in binary.
When $\ell$ fits in one code word, $a^\ell$ representation is still of constant size and
everything works smoothly.
For simplicity, for now we consider only this case, the general case is treated in Section~\ref{sec:improvements}.

\begin{algorithm}[H]
  \caption{\algremblocks{}: removing crossing blocks.
  \label{removing outer letters}}
  \begin{algorithmic}[1]
  \For{$i \gets 1 \twodots m+n$, except $n$ and $n+m$} \label{replace pop code begins}
  	\State let $X_i \to \alpha_i$ be the production for $X_i$ and $a$ its first letter
  	\State calculate and remove the $a$-prefix $a^{\ell_i}$ of $\alpha_i$
  	\label{replace pop code prefix}
	\State let $b$ be the last letter of $\alpha_i$ 
  	\State calculate and remove the $b$-suffix $b^{r_i}$ of $\alpha_i$
  	\label{replace pop code suffix}
  	\State replace each $X_i$ in rule's bodies by $a^{\ell_i}X_i b^{r_i}$
  	\If{$\eval(X_i) = \epsilon$} {} remove $X_i$ from the rules' bodies
  	\EndIf
	\EndFor \label{replace pop code ends}
	\end{algorithmic}
\end{algorithm}

After \algremblocks,
every letter $a$ has no crossing blocks and we may compress maximal blocks
using the already described method.
\begin{lemma}
\label{lem:removing crossing blocks}
After \algremblocks{} there are no crossing blocks.
This algorithm and following block compression can be performed in time $\Ocomp(|G| + (m+n)\log(m+n))$
and introduce at most $4$ new letters to each rule.
\end{lemma}
\begin{proof}
We first show the first claim of the lemma, i.e.\ that after
\algremblocks{} there are no letters with crossing blocks.
This follows from three observations:
\begin{enumerate}
	\item \label{obs 0}
	When \algremblocks{} considers $X_i$ with a rule $X_i \to \alpha_i$
	such that $\eval(X_i) = a^r w b^\ell$, where $w$ does not start with $a$ and
	does not end with $b$, then $\alpha_i$ has an explicit $a^\ell$ prefix and explicit $b^r$ suffix.
	\item \label{obs 1}
	When \algremblocks{} replaces $X_i$ with $a^{\ell_i}X_ib^{r_i}$ then afterwards the only letter to the left (right) of $X_i$
	in the rules is $a$ ($b$, respectively).
	\item \label{obs 2}
	After \algremblocks{} considered $X_i$, and $X_i$ is to the left (right)
	of $a$ then $a$ is not the first (last, respectively) letter of $\eval(X_i)$.
\end{enumerate}
All properties follow by a simple induction on the number $i$ of considered nonterminal.\

We infer from these observations that after \algremblocks{} there are 
no crossing blocks in $G$.
Suppose for the sake of contradiction, that there are;
let $a$ be the letter that has a crossing block.
By symmetry we consider only the case, when there are $X_i$ and $X_j$ such that
$aX_j$ appear in the rule for $X_i$ and $\eval(X_j)$ begins with $a$.
Note that by observation~\ref{obs 0} when \algremblocks{} considered $X_j$ then it replaced it with $b^\ell X_jc^r$
for some letters $b$ and $c$.
By observation~\ref{obs 1} the letter to the left of $X_j$ in the rule for $X_i$ is not changed by \algremblocks{}
afterwards (except that it can be popped when considering $X_i$) hence $b = a$.
Lastly, by observation~\ref{obs 2} the first letter of $\eval(X_j)$ is not $a$, contradiction.

\algremblocks{} is performed in $\Ocomp(|G|)$ time:
assuming that we represent block $a^\ell$ as a pair $(a,\ell)$,
the length of the $a$-prefix ($b$-suffix) is calculated simply by reading the rule
until a different letter is read (note that the lengths of the blocks fit in one machine word).
Since there are at most $4$ symbols introduced by \algremblocks{}
to the rule, this takes at most $\Ocomp(|G|)$ time.
The replacement of $X_i$ by $a^{\ell_i}X_ib^{r_i}$ is done at most twice
inside one rule and so takes in total $\Ocomp(n+m)$ time.

Note that right after \algremblocks{} it might be
that there are neighbouring blocks of the same letter in the rules of $G$.
However, we can easily replace such neighbouring blocks by one block of
appropriate length by in one reading of $G$, in time $\Ocomp(|G|)$.

Concerning the compression of the blocks of letters,
we adapt the block compression from~\algSET, see Lemma~\ref{lem:SET linear},
it is done in a similar way as we adapted the compression of non-crossing pairs from \algSET, 
see Lemma~\ref{lem:noncrossing and inner}.
For the sake of completeness, we present a sketch:
We read the description of $G$. Whenever we spot a maximal block $a^\ell$
for some letter $a$, we add a triple $(a,\ell,p)$ to the list.
The $p$ is the pointer to this appearance of the block in $G$.
Notice, that as there are no crossing blocks,
the nonterminals (and end or rules) count for termination of maximal blocks.

After reading the whole $G$ we sort these pairs lexicographically.
However, we sort separately the blocks that include the $a$-prefixes (or $b$-suffixes)
popped from nonterminals and the other blocks.
As in total we popped at most $4(n+m)$ prefixes and suffixes, there are at most $4(n+m)$
blocks of the former form, so we can sort their tuples in $\Ocomp((n+m) \log (n+m))$ time,
using any usual sorting algorithm of.
The remaining blocks are sorted using \algradix{} in linear time:
note that other blocks cannot have length greater than $|G|$,
and as $\Sigma = \Ocomp((n+m)\log M \log(n+m)) = \Ocomp((n+m)^3)$,
those tuples can be sorted in $\Ocomp(|G|)$ time.
Lastly, we can merge those two lists in $\Ocomp(|G|)$ time.

Now, for a fixed letter $a$, we use the pointers to localise
$a$'s blocks in the rules and we replace each of its maximal block 
of length $\ell>1$ by a fresh letter.
Since the blocks of $a$ are sorted according to their length,
all blocks of the same length are consecutive on the list,
and replacing them by the same letter is easily done.

Since we already know that there are no letters with crossing block,
we can show, as in Lemma~\ref{lem:noncrossing and inner},
that this procedure realises the block compression. The simple proof,
which is essentially the same as the proof in Lemma~\ref{lem:noncrossing and inner},
is omitted.
\qed
\end{proof}

\subsection{Grammar and alphabet sizes}
The subroutines of \algFCPM{} run in time dependant on $|G|$
and $|\letters|$, we bound these sizes.

\begin{lemma}
\label{lem:running time}
During \algFCPM, $|G| = \Ocomp((n+m)\log(n+m))$
and $|\letters| = \Ocomp((n+m) \log (n+m) \log|M|)$.
\end{lemma}
The proof is straightforward: using an argument similar to Lemma~\ref{lem:SET shortens}
we show that the size of the words that were in a rule at the beginning of the phase
shorten by a constant factor (in this phase).
On the other hand, only \algpop{} and \algremblocks{}
introduce new letters to the rules and it can be estimated,
that in total they introduced $\Ocomp(\log(n+m))$ letters to a rule in each phase.
Thus, bound $\Ocomp(\log(n+m))$ on each rules' length holds.
Concerning $|\letters|$, new letters appear as a result of a compression.
Since each compression decreases the size of $|G|$ by at least $1$,
there are no more than $|G|$ of them in a phase, which yields the bound.

\begin{proof}
We begin with showing the bound on $|G|$.
Consider a rule of $G$. On one hand, its size drops, as we compress letters in it. 
On the other, some new letters are introduced to the rule,
by popping them from nonterminals.
We estimate both influences.

Observe that Claim~\ref{clm:two letters are compressed}
applies to the bodies of the rules and so an argument similar
to the one in the proof Lemma~\ref{lem:SET shortens}
can be used to show that the length of the explicit strings that were in the rules at the beginning of the phase
decreases by a constant factor in each phase.
Of course, the newly introduced letters may be unaffected by this compression.
By routine calculations, as each rules' length decrease by a constant factor,
if $\Ocomp((n+m)\log(n+m))$ letters are introduced to $G$,
the $|G|$ is also $\Ocomp((n+m)\log(n+m))$ (with a larger constant, though).
Hence it is left to show that $\Ocomp((n+m)\log(n+m))$ letters are introduced to $G$ in one phase.
We do not count the letters that merely replaced some other letters (as a compression
of maximal block or a pair compression), but only the letters
that were popped into the rules.

In noncrossing pair compression there are no new letters introduced.
Concerning the crossing pairs compression, by Lemma~\ref{lem:compression of all pairs}
this introduces at most $\Ocomp(\log(n+m))$ letters to a rule, which is fine.
When \algremblocks{} is applied,
it introduces at most $4$ new symbols into a rule, see Lemma~\ref{lem:removing crossing blocks}.
In total, this gives $\Ocomp(\log(n+m))$ letters per rule, so $\Ocomp((n+m)\log(n+m))$ letters in total.

Concerning the alphabet, the time used in one phase,
is $\Ocomp((n+m) \log (n+m) + |G|)$, which is $\Ocomp((n+m) \log (n+m))$.
Thus no more than this amount of letters is introduced in one phase.
Lemma~\ref{lem:SET shortens} guarantees that there are $\Ocomp(\log M)$ phases,
and so a bound $\Ocomp((n+m) \log M \log (n+m))$ on $|\letters|$ follows.
\qed
\end{proof}

\subsubsection*{Memory Consumption}
\algET{} uses memory proportional to the size of grammar representation,
so $\Ocomp((n+m) \log (n+m))$ space, assuming that numbers up to $M$ fit in $\Ocomp(1)$ code words.

\subsubsection*{Main proof}
\label{subsec:sketch}
The cost of one phase of \algET{} is
$\Ocomp(|G| + (n+m) + (m+n)\log (n+m))$,
by Lemmas~\ref{lem:noncrossing and inner},
\ref{lem:compression of all pairs} and 
\ref{lem:removing crossing blocks}
while Lemma~\ref{lem:running time} shows that $|G| = \Ocomp((n+m)\log(n+m))$
and Lemma~\ref{lem:SET shortens} shows that there are $\Ocomp(\log M)$ phases.
So the total running time is $\Ocomp((n+m) \log M \log(n+m))$.

\section{Pattern matching}

In Section~\ref{sec: first and last} it was shown how to perform the pattern matching using recompression on explicit strings.
In this section we extend this method to the case in which \pattern{} and \mytext{} are given as SLPs.
Note that most of the tools are already known, as in Section~\ref{sec:outline}
it was shown how to perform the equality testing when both \pattern{} and \mytext{} are given as SLPs.
In particular, the proof of correctness of the pattern matching follows from the one in Section~\ref{sec: first and last},
so we need to focus only on the efficient implementations, mostly of \algfixdiff{} and \algfixsim,
as other operations are used already in \algET.

The outline of the algorithm looks as follows, in the rest of the section we comment on the implementation details and running time.
\begin{algorithm}[H]
	\caption{\algFCPM: outline}
	\label{alg:main}
	\begin{algorithmic}[1]
	\While{$|\pattern|>1$}
		\State $P \gets $ list of pairs, $L \gets$ list of letters
		\State fix the beginning and end \Comment{See Section~\ref{sec: first and last}}
    \For{each $a \in L$} {} compress blocks of $a$
		\EndFor
		\State $P' \gets $ crossing pairs from $P$, $P \gets$ noncrossing pairs from $P$
		\For{each $ab\in P$} {} compress pair $ab$ 
		\EndFor
		\For{$ab \in P'$} {} compress pair $ab$
		\EndFor
		\EndWhile
	\State Output the answer.
	 \end{algorithmic}
\end{algorithm}

The first operation in the \algFCPM{} is the fixing of the beginning and end,
which adapts \algfixsim{} and \algfixdiff{} to the compressed setting.

\begin{lemma}
The fixing of beginning and end for an SLP represented \pattern{} and \mytext{}
can be performed in $\Ocomp(|G| + (n+m)\log (n+m))$ time.

It introduces $\Ocomp(n+m)$ new letters to $G$.
\end{lemma}
\begin{proof}
To see this, we look at the operations performed by \algfixsim{} (the ones for \algfixdiff{} are even simpler)
and comment how to perform them efficiently.
Firstly, in linear time we can find out what is the first and last letter of \pattern,
to see whether \algfixdiff{} or \algfixsim{} should be applied, suppose the latter.
Now \algfixsim{} performs a (modified) block compression, the only difference is that we compress only
blocks of $a$ and replace them not by a single letter, but by up to three letters.
To this end we apply a modified \algremblocks, which removes only $a$-prefixes and $a$-suffixes
and afterwards compress only blocks of $a$.
The running time bounds $\Ocomp(|G| + (n+m)\log(n+m))$, see Lemma~\ref{lem:removing crossing blocks},
are preserved, furthermore, using the same argument as in Lemma~\ref{lem:removing crossing blocks}
it can be shown that after the modified block compression there are no crossing $a$ blocks.
Furthermore, by the same lemma $\Ocomp(n+m)$ new letters are introduced to $G$.

The next operations in \algfixsim{} is the compression of pairs of the form
$\makeset{a_L b}{b \in \Sigma \setminus a_L}$, then $\makeset{b a_R}{b \in \Sigma \setminus a_R}$
(and then perhaps also $\makeset{a_1 b}{b \in \Sigma \setminus a_1}$).
In each case the pairs are obtained by partitioning the alphabet into $\Sigma_\ell$ and $\Sigma_r$,
(where one of the parts is a singleton).
Thus by Lemma~\ref{lem:uncrossing pairs} one such group can be uncrossed in $\Ocomp(n+m)$ time,
the uncrossing introduces $\Ocomp(n+m)$ letters to $G$.
Afterwards we can compress all pairs by naively in $\Ocomp(|G|)$ time.
\qed
\end{proof}

The rest of the operations on $G$ (pair compression, block compression) is implemented as in Section~\ref{sec:outline}
and has the same running time.

We now move to the analysis of \algFCPM. We show that \algFCPM{} preserves the crucial important property of \algET:
that $|\pattern|$ decreases by a constant factor in each phase and that $|G| = \Ocomp((n+m)\log(n+m))$.

\begin{lemma}
\label{lem:FCPM shortening}
In each phase the \algFCPM{} shortens \pattern{} by a constant factor.
The size of $G$ is $\Ocomp((n+m)\log(n+m))$, while the size of $\letters$ is $\Ocomp((n+m)\log(n+m)\log M)$
\end{lemma}
\begin{proof}
Observe that \algFCPM{} performs the same operations on \pattern{} as \algSPM,
but it just does it on the compressed representation.
Thus it follows from Lemma~\ref{lem:fixing shortens} that both \pattern{} and \mytext{} are shortened by a constant factor in one phase of \algFCPM.

Concerning the size of the grammar, a similar argument as in Lemma~\ref{lem:running time} applies:
note that as \algET{} the \algFCPM{} introduces $\Ocomp((n+m)\log(n+m))$ letters per phase into $G$.
On the other hand, the analysis performed in Lemma~\ref{lem:fixing shortens} (that \algSPM{} shortens \pattern)
applies to each substring of \pattern{} and \mytext, so each explicit string in the rules of $G$
is shortened during the phase by a constant factor, i.e.\ the same as in Lemma~\ref{lem:running time}.
Hence the size of $G$ kept by \algFCPM{} can be also bounded by $\Ocomp((n+m)\log(n+m))$.
Consequently, also $|\letters| = \Ocomp((n+m)\log(n+m)\log M)$.
\qed
\end{proof}

Now, Lemma~\ref{lem:FCPM shortening} implies that \algFCPM{} runs in $\Ocomp((n+m)\log(n+m) \log M)$ time:
each subprocedure runs in time $\Ocomp(|G| + (n+m)\log(n+m)) = \Ocomp((n+m)\log(n+m))$ and so this is also the running time
of one phase.
Since pattern is shortened by a constant factor in each phase, see again Lemma~\ref{lem:FCPM shortening},
there are $\Ocomp(\log M)$ many phases.
The correctness (returning representation of all pattern appearance) follows from the correctness of \algSPM{}
(as the performed operations are the same, just the representation of \pattern{} and \mytext{} is different).

\begin{theorem}
\algFCPM{} runs in $\Ocomp((n+m)\log(n+m) \log M)$ time and returns a representation of all pattern appearances
in text.
\end{theorem}

\subsubsection*{Positions of appearances}
In order to give the position of the first appearance of the pattern we need to track to how many letters in the input
the current symbol of $\Sigma$ corresponds.
This is formalised using \emph{weight} of letters, which is extended to strings in a natural way:
Every letter $a$ in the input grammar has $\weight(a) = 1$,
while when a new letter $a$ replaces the string $w$ we set $\weight(a) = \weight(w)$.
When $N$ fits in a constant amount of code words,
the weight of each letter can be calculated in constant time, so we can store the weights of the letters on the fly in a table.

Since the compression can be seen as building an SLP for the input,
the weights of the letters are well defined (recall that we can always imagine new letters replace non-empty strings in the instance,
see Section~\ref{sec:toy example} and end of Section~\ref{sec: first and last}).
Thus, to calculate the position of the first pattern appearance it is enough to calculate the weight of the string preceding it.
To this end we keep up-to-date table of weights of $\eval(X_i)$, for each $X_i$.
To return the first position of a pattern appearance we determine the derivation path
for this appearance and sum up the weights of nonterminals and letters that are to the left
of this derivation path; this is easy to perform in time linear in $|G|$.

Note that there is a small technical issue: in one special case we remove the first letter from \mytext, when it is $a_R$.
But $\weight(a_R) = 0$ and so it does not influence anything.
When considering the last appearance of the pattern, note that the $a_L$ that is removed from the end
has a non-negative weight,
still it is enough to add the weights of all letters removed from the end of \mytext.

\section{Improving running time}
\label{sec:improvements}
In order to reduce the running time to $\Ocomp((n+m)\log M)$ we need to make sure that the grammar size
is $\Ocomp(n+m)$ and improve the running time of block compression, see Lemma~\ref{lem:removing crossing blocks},
so that it is $\Ocomp(|G|)$, without the extra $(n+m)\log(n+m)$ summand.
For the former, the argument in Lemma~\ref{lem:running time} (and its adaptation in Lemma~\ref{lem:FCPM shortening})
 guarantee this size as long as there are only $\Ocomp(n+m)$ letters introduced to $G$ in a phase.
The crossing blocks compression already posses this property, see Lemma~\ref{lem:removing crossing blocks},
so it is enough to alter the crossing pairs compression.

We show that it is enough to consider $\Ocomp(1)$ partitions $\letters_\ell,\letters_r$ and pairs that fall into them.
Roughly, we choose a partition such that a constant fraction of crossing pairs appearances in \pattern{} fall into this partition.
In particular, we calculate for each crossing pair $ab$ the number of its appearances in \pattern,
so we need to manipulate numbers up to $M$ in constant time,
i.e.\ this construction requires that $M$ fits in $\Ocomp(1)$ code words.

For the block compression, we improve the sorting time: we group block lengths into groups of similar lengths
and sort them inside one such group in linear time using \algradix.
The groups are also established using \algradix{} performed on representatives of groups.
The latter sorting treats numbers as bit string and therefore may have high
running time, but we show that overall it cannot exceed $\Ocomp((n+m)\log M)$ during the whole \algFCPM.

\subsection{Faster compression of crossing pairs}
Let us formalise the notion that $ab$ falls into a partition of \letters:
for a given partition $\letters_\ell,\letters_r$ we say that it \emph{covers} the appearances of $ab \in \letters_\ell\letters_r$ in \pattern.
The main idea of improving the running time of the crossing pairs is quite simple:
instead of considering $\Ocomp(\log(n+m))$ partitions such that each crossing pair from $P'$
is covered by at least one of them, we consider only one partition $\letters_\ell,\letters_r$
such that at least one fourth of appearances of crossing pairs in \pattern{} are covered by it.
Then estimations about shortening of the pattern in one phase hold as before, though with a larger constant.

Existence of such a partition can be shown by a simple probabilistic argument: if we assign each letter
to $\letters_\ell$ with probability $1/2$ then a fixed appearances of $ab$ in \pattern{} is covered with probability $1/4$.
The standard expected-value derandomisation technique gives a deterministic algorithm finding such a partition,
it can be easily implemented in $\Ocomp(|G|)$ time, see Lemma~\ref{lem:finding partition}.

It is not guaranteed that this partition shortens also $|G|$,
however, we can use exactly the same approach to shorten $G$:
we find another partition $\letters_\ell \letters_r$ such that at least $1/4$ of crossing pairs explicit appearances 
in $G$ are from this partition.

Our to-be-presented algorithm constructing a partition requires a list of all crossing pairs,
together with number of their appearances in \pattern.
This can be supplied using a simple linear-time algorithm:
for each nonterminal $X_i$ we calculate the amount $k_i$ of substrings $\eval(X_i)$ it generates in \pattern.
We associate an appearance of $ab$ with the least nonterminal that generated it.
Then the number of all appearances of $ab$ can be calculated summing appropriate $k_i$s.

\begin{lemma}
\label{lem:number of appearances}
Assuming $M$ fits in $\Ocomp(1)$ code words,
in $\Ocomp(|G| + n + m)$ we can return a sorted list of crossing pairs together with number of their appearances in \pattern.
\end{lemma}
\begin{proof}
Clearly $k_m = 1$ (as $X_m$ simply generates the whole \pattern) and other numbers satisfy a simple recursive
formula:
\begin{equation}
\label{eq:number of substrings}
k_j = \sum_{i > j} k_i \cdot \#\{\text{number of times } X_j \text{ appears in the rule for } X_i\} \enspace .
\end{equation}
Then~\eqref{eq:number of substrings} can be used in a simple linear-time procedure for calculation of $k$s:
for $i = m \twodots 1$ we add $k_i$ to $k_j$, when $X_j$ appears in the rule for $X_i$
(we add twice if there are two such appearances).
Clearly this can be implemented in linear time.

Concerning the number of appearances of crossing pairs,
observe that each appearance of $ab$ in \pattern{} can be assigned to a unique rule $X_i \to \alpha_i$:
this is the rule that generates this particular appearance of $ab$ and moreover this appearance of $ab$
comes from an explicit appearance of $ab$ in $\alpha_i$ or a crossing appearance of $ab$ in this rule.
To see this imagine we try to retrace the generation of this particular $ab$:
Given $X_i$ generating this appearance of $ab$
(we start with \patternc, as we know that it generates this $ab$)
we check if it is generated by nonterminal $X_j$ in the rule.
If so, we replace $X_i$ with $X_j$ and iterate the process.
If not, then this $ab$ is comes from either an explicit or crossing pair in this $X_i$.

Given a rule for $X_i$ listing all pairs that appear explicitly or have a crossing appearance in the rule
for $X_i$ is easy,
for each such pair $ab$ we create a tuple $(a,b,k_i)$ (where $k_i$ is the number of substrings that $X_i$ generates).
We sort the tuples using \algradix{} (in $\Ocomp(|G| + n + m)$ time).
Now for a given pair $ab$ the tuples with number of its appearances are listed consecutively on the list,
so for each pair we can add those numbers and obtain the desired (sorted) list of pairs with numbers of their
appearances in $G$, this also takes linear time, since list is of this length.

This list includes both crossing and non-crossing pairs.
We use the same procedure as in Lemma~\ref{lem:noncrossing and inner} to establish the crossing and non-crossing pairs.
Note that it generated a sorted list of crossing (and non-crossing) pairs, this takes $\Ocomp(|G| + n + m)$ time.
Without loss of generality, the order on those lists is the same as on our list,
so we can filter from it only crossing pairs in linear time.
\qed
\end{proof}

In the following, for a crossing pair $ab$ we shall denote by $k_{ab}$ the number of its appearances in \pattern,
calculated in Lemma~\ref{lem:number of appearances}.

Now we are ready to find the partition covering at least one fourth of the appearances of crossing pairs
is done by a derandomisation of a probabilistic argument showing its existance:
divide $\Sigma$ into $\Sigma_\ell$ and $\Sigma_r$ randomly, where each letter
goes to each of the parts with probability $1/2$.
Consider an appearance of a crossing pair $ab$ in \pattern.
Then $a \in \Sigma_\ell$ and $b \in \Sigma_r$ with probability $1/4$.
This applies to every appearance of a crossing pair in \pattern,
so the expected number of pairs covered is $1/4$ of their number.

\begin{lemma}[cf.~\cite{grammar}]
\label{lem:finding partition}
For \pattern{} in $\Ocomp(|G| + n + m)$ time we can find
a partition of $\Sigma$ into $\Sigma_\ell$, $\Sigma_r$ such that
number of appearances of crossing pairs in \pattern{} covered by this partition
is at least $1/4$ of all such appearances in \pattern.
In the same running time we can provide for each covered crossing $ab$
a lists of pointers to its explicit appearances in $G$.
\end{lemma}
\begin{proof}
Observe first that the above probabilistic argument can be altered:
if we were to count the number of pairs that are covered \emph{either} by $\Sigma_\ell\Sigma_r$
or by $\Sigma_r\Sigma_\ell$ then the expected number of crossing pairs appearances covered by
$\Sigma_\ell \Sigma_r \cup \Sigma_r \Sigma_\ell$ is one half.

The deterministic construction of such a partition follows by a simple derandomisation,
using an expected value approach.
It is easier to first find a partition such that at least half
of crossing pairs' appearances in \pattern{} are covered
by $\Sigma_\ell\Sigma_r \cup \Sigma_r\Sigma_\ell$, we then choose $\Sigma_\ell\Sigma_r$ or $\Sigma_r\Sigma_\ell$,
depending on which of them covers more appearances.

According to Lemma~\ref{lem:number of appearances} we assume that we are given a sorted list $P'$,
on which we have all crossing pairs together with the number $k_{ab}$ of their appearances in \pattern.

Suppose that we have already assigned some letters to $\Sigma_\ell$ and $\Sigma_r$
and we are to decide, where the next letter $a$ is assigned.
If it is assigned to $\Sigma_\ell$, then all appearances of pairs from $a\Sigma_\ell \cup \Sigma_\ell a$ are not going
to be covered, while appearances of pairs from $a \Sigma_r \cup \Sigma_r a$ are;
similarly observation holds for $a$ being assigned to $\Sigma_r$.
The algorithm makes a greedy choice, maximising the number of covered pairs in each step.
As there are only two options, the choice brings in at least half of appearances considered.
Lastly, as each appearance of a pair $ab$ from \pattern{} is considered exactly once
(i.e.\ when the second of $a$, $b$ is considered in the main loop),
this procedure guarantees that at least half of appearances of crossing pairs in \pattern{} is covered.

In order to make the selection effective,
the algorithm \alggreedypairs{} keeps an up-to-date counters
$\mycount_\ell[a]$ and $\mycount_r[a]$, 
denoting, respectively, the number of appearances of pairs from
$a\Sigma_\ell \cup \Sigma_\ell a$ and $a\Sigma_r \cup \Sigma_r a$ in \pattern.
Those counters are updated as soon as a letter is assigned to $\Sigma_\ell$ or $\Sigma_r$.
Note that as by Claim~\ref{clm:consecutive values} we can assume that letters in \pattern{} are from an interval of consecutive $|G|$ letters,
this can be organised as a table with constant access time to $\mycount_\ell[a]$ and $\mycount_r[a]$.

\begin{algorithm}[H]
  \caption{\alggreedypairs{} \label{greedy pairs}}
  \begin{algorithmic}[1]
  \State $L \gets $ set of letters used in $P'$
	\State $\Sigma_\ell \gets \Sigma_r \gets \emptyset$ \Comment{Organised as a bit vector}
	\For{$a \in L$}
		\State $\mycount_\ell[a] \gets \mycount_r[a] \gets 0$ \Comment{Initialisation}
	\EndFor
	\For{$a \in L$}
		\If{$\mycount_r[a] \geq \mycount_\ell[a]$} \Comment{Choose the one that guarantees larger cover}
			\State $\mychoice \gets \ell$ 
		\Else
			\State $\mychoice \gets r$
		\EndIf
		\State $\Sigma_{\mychoice} \gets \Sigma_{\mychoice} \cup \{a\}$\
		\For{each $b \in L$}
			\State $\mycount_{\mychoice}[b] \gets \mycount_{\mychoice}[b] + k_{ab} + k_{ba}$ \label{counter update}
		\EndFor
	\EndFor
	\If{\# appearances of pairs from $\Sigma_r\Sigma_\ell$ in \pattern > \# appearances of pairs from $\Sigma_\ell\Sigma_r$ in \pattern}
		\State switch $\Sigma_r$ and $\Sigma_\ell$ \label{actual partition}
	\EndIf
	\State \Return $(\Sigma_\ell, \Sigma_r)$
  \end{algorithmic}
\end{algorithm}

By the argument given above, when $\Sigma$ is partitioned into $\Sigma_\ell$ and $\Sigma_r$,
at least half of the appearances of pairs from \pattern{} are covered by
$\Sigma_\ell\Sigma_r \cup \Sigma_r \Sigma_\ell$.
Then one of the choices $\Sigma_\ell\Sigma_r$ or $\Sigma_r\Sigma_\ell$ covers at least one fourth of the appearances.

It is left to give an efficient variant of \alggreedypairs,
the non-obvious operations are the choice of the actual partition in line~\ref{actual partition}
and the updating of $\mycount_\ell[b]$ or $\mycount_r[b]$ in line~\ref{counter update}. 
All other operation clearly take at most $\Ocomp(|G| + n + m)$ time.
The latter is simple: since $\Sigma_\ell$ and $\Sigma_r$ as organised as a bit vector,
we can read $P'$, for each pair in it check if it is covered by $\letters_\ell\letters_r$ or $\letters_r\letters_\ell$
and calculate the total number of pairs appearances covered by each of those two partitions.

To implement the $\mycount$,
for each letter $a$ in \pattern{} we have a table \rightl[] of \emph{right lists}:
$\rightl = \makeset{(b,k_{ab})}{ab \text{ appears in } P'}$, represented as a list.
There is a similar \emph{left list} $\leftl = \makeset{(b, k_{ba})}{ba \text{ appears in } P'}$.
Since at the input we get a sorted list of all pairs $ab$ together with $k_{ab}$,
creation of $\rightl$ can be easily done in in linear time (and similarly $\leftl$ can).

Given \rightl[] and \leftl[], performing the update in line~\ref{counter update} is easy
(suppose that we are to update $\mycount_\ell$):
we go through $\rightl$ ($\leftl$) and increase the $\mycount_\ell[b]$ by $k_{ab}$ ($k_{ba}$, respectively).
As \rightl[], \leftl[] and \mycount{} are organised as tables,
this takes only $\Ocomp(1)$ per read element of  $\rightl$ ($\leftl$).
We can then discard  $\rightl$ ($\leftl$) as they are not going to be used again.
In this way each of the list $\rightl$ ($\leftl$) is read $\Ocomp(1)$ times during \alggreedypairs,
and so this time is at most as much as the time of their creation, i.e.\ $\Ocomp(|G|)$.
\qed
\end{proof}

A similar construction works also when we want to calculate the partition
that covers $1/4$ of appearances of crossing pairs in $G$:
when calculating the number of appearances of pair $ab$ it is enough to drop the coefficient $k_i$ for appearing in the rule $X_i$
and take $1$ for every rule. The rest of the construction and proofs is the same.

\begin{lemma}
\label{lem:finding partition 2}
In $\Ocomp(|G| + n + m)$ time we can find
a partition of $\Sigma$ into $\Sigma_\ell$, $\Sigma_r$ such that
number of appearances of crossing pairs in $G$ covered by this partition
is at least $1/4$ of all such appearances in $G$.
In the same running time we can provide for each covered crossing $ab$
a lists of pointers to its explicit appearances in $G$.
\end{lemma}

Thus the modification to \algFCPM{} would be as follows:
after establishing the list of all crossing pairs we find two partitions $\letters_\ell,\letters_r$,
$\letters_\ell',\letters_r'$ one of which covers half of appearances of crossing pairs in the pattern the other in $G$.
And then instead of compressing all crossing pairs
we compress only those covered by the first and then by the second of those two partitions.
Each of those compression requires only one call to \algpop, so there are only $\Ocomp(1)$
letters introduced to a rule during the crossing pairs compression.

\begin{lemma}
\label{lem:pair comp improved}
\algFCPM{} using the modified crossing pair subprocedure introduces $\Ocomp(1)$
letters to a rule in one phase.
\end{lemma}

It is left to estimate that indeed this modified compression schema shortens $|\pattern|$ and $|G|$ by a constant factor in a phase.
This will clearly guarantee the $\Ocomp(\log M)$ number of phases.
\begin{lemma}
\algFCPM{} using the modified crossing pair subprocedure keeps the size of the grammar $\Ocomp(n+m)$
and has $\Ocomp(\log M)$ phases
\end{lemma}
\begin{proof}
Let us first consider the simpler case of \algET.
Consider first the length of \pattern, we show that it is reduced by a constant factor in a phase.
Consider two consecutive letters $ab$ in \pattern.
Observe that if \algET{} tried to compress $ab$ (when $a = b$ this means that $a$ blocks were compressed)
then at least one of those letters is compressed in the phase:
we tried to compress this $ab$ and the only reason why we could fail is because one of those letters
was already compressed by some earlier compression.
We want to show that for at least $1/4$ of all such pairs $ab$ we tried to compress them.
Thus at least $1/8$ of all letters was compressed and so the length of \pattern{} dropped by at least $1/16$ in a phase.

If $a = b$ then we compressed them, during the blocks compression.
If $a \neq b$ and $ab$ is non-crossing then we tried to compress them.
Lastly, when $a \neq b$ and $ab$ is a crossing pair, then we chose a partition $\Sigma_\ell \Sigma_r$ such that
at least $1/4$ of all appearances of crossing pairs is covered by this partition.
So for one in four of such pairs we tried to compress them.

In total, for at least one fourth of $ab$s we tried to compress them, as claimed.

A similar analysis yields that we reduced the length of $|G|$ (excluding the new introduced letters) by $15/16$.
Since we introduce only $\Ocomp(n+m)$ new letters to $G$ per phase, the size of $G$ remains $\Ocomp(n+m)$.

Now, in the general case of \algFCPM{} we combine this analysis with the one in Lemma~\ref{lem:fixing shortens}.
We define the fragments of more than one letter as in Lemma~\ref{lem:fixing shortens},
i.e.\ the letters that were replaced during the compression
are grouped so that one group (fragment) is replaced with a shorter string.
The letters that were not altered are not assigned to any fragments.

Similarly as in the earlier argument for \algET{} above,
we want to show that for at least $1/4$ of all pairs of consecutive letters
one of those letters was assigned to a fragment.
Since fragments are replaced with strings of at most $3/4$ of their length,
as above this shows that \pattern{} is shortened by a constant factor.
To show that at least one of $ab$ is in a compressed fragment it is again enough to show
that we tried to compress $ab$ (either as a pair of different letters or as a part of a block of letters):
if we succeed then $a$, $b$ are in the same fragment, if we fail then this means that at least one of them
is in some other fragment.

So consider any two consecutive letters $a$ and $b$.
If any of them was compressed during the fixing of beginning or end then we are done,
as it was assigned to a fragment.
Otherwise, if $a = b$ than they are compressed during the blocks compression,
so both of them are assigned to the same fragment.
If $a \neq b$ and $ab$ is a non-crossing pair, then we tried to compress it during the non-crossing pairs compression.
Lastly, if $a \neq b$ and $ab$ is a crossing pair then due to our construction of $\Sigma_\ell$ and $\Sigma_r$
from Lemma~\ref{lem:finding partition 2} at least one fourth of appearances of crossing pairs is chosen for the compression.

The rest of the argument follows as in the case of the one for \algET, with a slightly larger constant.
Hence, in each round \pattern{} is shortened by a constant factor and so there are at most $\log M$ phases.

Observe that a similar argument holds for $G$: there is a second round of compression of crossing pairs that
tries to compresses at least $1/4$ of crossing pairs appearances in $G$.
Hence also the explicit strings in $G$ can be grouped into fragments as above.
On the other hand, by Lemma~\ref{lem:removing crossing blocks} and~\ref{lem:pair comp improved}
there are $\Ocomp(n + m)$ letters introduced to $G$ in one phase
(and those are not necessarily compressed).
So the size of the new grammar (at the end of the phase) $G'$ can be given using an recursive equation
$$
|G'| \leq \alpha |G| + \beta (n+m)
$$
for some $\alpha < 1$ and $\beta$.
Since in the first phase $|G| = 2(n+m)$ by simple calculations
it follows that $|G'| \leq \frac{2\beta}{1-\alpha} (n+m)$.
\qed
\end{proof}

\subsection{Block compression}
As already noted, we should improve the $\Ocomp(|G| + (n+m) \log(n+m))$ running time,
see Lemma~\ref{lem:removing crossing blocks} used for sorting of blocks' lengths to $\Ocomp(|G|)$.
We deal with this by introducing a special representation of the lengths of $a$ blocks.
In particular, we shall sort the lengths of blocks using \algradix, treating the lengths as bitvectors.
For this reason considering very long blocks that exceed the length of \pattern{} needs to be avoided.

\subsubsection*{Too long blocks}
Consider the blocks of letter $a$ that does not appear in \pattern.
Then there is no difference, whether we replace two appearances of $a^\ell$
with the same letter, or with different letters, as they cannot be part of a pattern appearance.
Thus, for $a$ that does not appear in \pattern{} we perform a `sloppy' blocks compression:
we treat each maximal block as if it had a unique length.
To be precise: we perform \algremblocks, but represent $a^\ell$ blocks as $(a,?)$ for $\ell>1$.
Then, when replacing blocks of $a$ (we exclude the blocks of length $1$),
we replace each of them with a fresh letter.
In this way, the whole blocks compression does not include any cost of
sorting the lengths of blocks of $a$.
Still, the appearances of the pattern are preserved.

Similar situation appears for $a$ that appears in \pattern,
but \mytext{} has $a$ blocks of length greater than $M$.
We treat them similarly: as soon as we realise that $a^\ell$ has $\ell > M$,
we represent such blocks as $(a,>M)$ and do not calculate the exact length and do not insist that two
such blocks of the same length are replaced with the same symbol.
In this way we avoid the cost associated with sorting this length.
Of course, when $a$ is the first or last letter of the pattern we need to replace them with $a_Ra_?a_L$ (or similar),
to allow the pattern beginning/ending at this block.

\subsubsection*{Length representations}
The intuition is as follows: while the $a$ blocks can have exponential length,
most of them do not differ much, as they are obtained by concatenating letters $a$ that appear explicitly in the grammar.
Such concatenations can in total increase the lengths of $a$ blocks by $|G|$.
Still, there are blocks of exponential length:
these `long' blocks are created only when two blocks coming from two different
non-terminals are concatenated. However, there are only $n+m$ concatenations of nonterminals,
and so the total number of `long' blocks `should be' at most $n+m$.
Of course, the two mentioned ways of obtaining blocks can mix, and our representation
takes this into the account:
we represent each block as a concatenation of two blocks: `long' one and `short' one:
\begin{itemize}
	\item the `long' corresponds to a block obtained as a concatenation
	of two nonterminals, such a long block is common for many blocks of letters,
	\item the `short' one corresponds to concatenations of letters
	appearing explicitly in $G$, this length is associated with the given block alone.	
\end{itemize}

More formally:
we store a list of \emph{common lengths},
i.e.\ the lengths of common long blocks of letters.
Each block-length $\ell$ is represented as a sum $c + o$,
where $c$ is one of the common lengths and $o$ (\emph{offset}) is a number associated with $\ell$.
Furthermore, some blocks are represented only by offsets;
we sometimes think of them as if they were represented by a common length $0$ and an offset.
The construction will guarantee that each offset is at most $|G|$.
Internally, $a^{\ell}$ is represented as a number $o$ and a pointer to $c$.

Initially a common length $c$ is created for each nonterminal $X_i$, such that $\eval(X_i) = a^c$.
Next, other common lengths are created, when we add two common lengths (perhaps with offsets),
i.e.\ when during the calculation of length $\ell$ (inside a rule) we add lengths that are both
represented using non-zero common lengths. This new length $\ell$ is then a new common length and is represented as itself plus a $0$ offset.
If we concatenate explicit letter $a$ (i.e.\ represented by a $0$-common length with an offset) to a block,
we simply increase the offset.
The blocks that are created solely by explicit letters $a$ are represented by offsets alone, without a common length.
Observe that this covers all possible way of creation of block.
Furthermore, there are at most $2(n+m)$ common lengths in one phase:
at most $n+m$ created when $X_i$ defines a block of letters and at most one per rule created as a concatenation of two
block whose lengths are represented as common lengths.

Before proceeding, let us note on how large the offsets may be and how many of them are.
\begin{lemma}
\label{lem:offsets}
There are at most $|G| + n + m$ offsets in total and largest offset is at most $|G|$.
\end{lemma}
\begin{proof}
Creation of an offset corresponds to an explicit letter in $G$, so there are at most $|G|$ offsets created.

An offset is created or increased, when an explicit letter $a$ (not in a compressed form) is concatenated to the block of $a$s.
One letter is used once for this purpose and there is no other way to increase an offset,
so the maximal of them is at most $|G|$.
\qed
\end{proof}

Since we intend to sort the lengths, we need to compare the lengths
of two numbers represented as common lengths with offsets, say $o + c$ and $o' + c'$.
Since the common lengths are so large, we expect that we can compare them lexicographically, i.e.\
\begin{equation}
\label{eq:comparing lex}
o + c \geq o' + c' \iff
\begin{cases}
c > c' , & \text{or }\\
c = c' \land o \geq o'
\end{cases}
\end{equation}
Furthermore~\eqref{eq:comparing lex} allows a simple way of sorting the lengths of maximal blocks:
\begin{itemize}
	\item we first sort the common lengths (by their values)
	\item then for each common length we (separately) sort the offsets assigned to this common length.
\end{itemize}
While~\eqref{eq:comparing lex} need not to be initially true,
we can apply a couple of patches which make it true.
Before that however, we need the common lengths to be sorted.
We sort them using \algradix{} and treating each common length as a series of bits.
Although this looks more expensive, it allows a nice amortised analysis as demonstrated later, see Lemma~\ref{clm:cost from power to rule}.
Recall that we do not sort lengths of blocks longer than $M$.

\begin{lemma}
\label{clm:cost of long sorting}
Let $c_1 \leq c_2 \leq \dots \leq c_k \leq M$ be the common lengths.
The time needed to sort them is $\Ocomp(\sum_{i=1}^k \log(c_i) + k)$.
\end{lemma}
This is done by a standard implementation of \algradix{}
that sorts the numbers of different length.

The problem with~\eqref{eq:comparing lex} is that even though $c_i$ and $c_j$ are so large,
it can still happen that $|c - c'|$ is small.
We fix this naively: first we remove common lengths so that $c_{i + 1} - c_i > |G|$.
A simple greedy algorithm does the job in linear time.
Since common lengths are removed, we need to change the representations of lengths:
when $o$ was assigned to removed $c$ consider the $c_i$ and $c_{i+1}$ that remained in the sequence and $c_i < c < c_{i+1}$.
We reassign $\ell = c + o$ to either $c_i$ or $c_{i+1}$:
if $o+c \geq c_{i+1}$ then we reassign it to $c_{i+1}$ and otherwise to $c_i$.
It can be shown that in this way all offsets are at most $2 |G|$ and that~\eqref{eq:comparing lex} holds afterwards.

\begin{lemma}
Given a sorted list of common lengths we can in $\Ocomp(|G|)$ time choose its sublist and reassign offsets
(preserving the represented lengths) such that all offsets are at most $2|G|$ and~\eqref{eq:comparing lex} holds.
\end{lemma}
\begin{proof}
Given a sorted list of of common lengths we choose a subsequence of it such that the distance between any two consecutive common lengths in it is at least $|G|$.
This is done naively: we choose $c_0 = 0$ and then go through the list.
Having last chosen $c$ we look for the smallest common length $c'$ such that $c' - c > |G|$ and choose this $c'$.
Since there are $2(n+m)$ common lengths in the beginning, this can be done in $\Ocomp(n+m)$ time.
We refer to the obtained sequence as $0 = c_0 < c_1 < \ldots $.

For any removed $c$ such that $c_i < c < c_{i+1}$
we reassign offsets assigned to $c$ as described above:
for $o$ assigned to $c$, if $c + o \geq c_{i+1}$ then we reassign $o$ to $c_{i+1}$, otherwise to $c_i$.
In this way $o$ is changed to $o'$ and it takes $\Ocomp(1)$ per offset to change this and as there are $\Ocomp(|G|)$
offsets, see Lemma~\ref{lem:offsets}, this takes $\Ocomp(|G|)$ time in total.
Let $o'$ be the offset after the reassignment. Then
\begin{itemize}
	\item $o' \leq 2|G|$, since $o \leq |G|$ and the only way to increase it is to reassign it to $c_i$.
	Since $c$ is removed, it holds that $c - c_i \leq |G|$. Hence $o' = o + (c - c_i) \leq |G| + |G|$.
	\item When $o_i$ is assigned to $c_i$ then $o_i + c_i < c_{i+1}$:
	indeed, if $o_i$ was reassigned from $c > c_i$ then by definition $c_i + o_i = c + o < c_{i+1}$;
	if $o$ was originally assigned to $c_i$ or it was reassigned from $c_{i-1}$ then $o_i < |G|$ and so $c_i + o_i \leq c_i + |G| < c_{i+1}$.
\end{itemize}
Hence the claim of the Lemma holds.
\qed
\end{proof}

Now, since~\eqref{eq:comparing lex} holds, in order to sort all lengths it is enough to sort the offsets within groups.
To save time, we do it simultaneously for all groups:
offset $o_j$ assigned to common length $c_i$ is represented as $(i,o_j)$,
we sort these pairs lexicographically, using \algradix.
Since the offsets are at most $2|G|$ and there are at most $|G|$ of them and there are at most $\Ocomp(n+m)$ common lengths,
\algradix{} sorts them in $\Ocomp(|G|)$ time.

\begin{lemma}
When all common lengths (not larger than $M$) are sorted and satisfy~\eqref{eq:comparing lex},
sorting all lengths takes $\Ocomp(|G|)$ time.
\end{lemma}

It is left to bound the sorting time of all common lengths.
Due to Lemma~\ref{clm:cost of long sorting} this cost is $\Ocomp(\log p)$ for a common length $c$.
We redirect the $\log(p)$ cost towards the rule, in which $c$ was created.
We estimate the total such cost over the whole run of \algFCPM.

\begin{lemma}
\label{clm:cost from power to rule}
For a single rule, the cost redirected from common lengths towards this rule
during the whole run of \algFCPM{} is $\Ocomp(\log M)$.
\end{lemma}
\begin{proof}
If a common length is created because some $X_j$ defined a block of $a$,
this costs at most $\log M$ and happens once per nonterminal, so takes $\Ocomp(\log M)$ time.

The creation of the common length can remove a nonterminal from the rule,
which happens when $X_j$ in the rule defines a string in $a^*$.
Then the cost is at most $\log M$ and such cost can be charged twice to a rule,
as initially there are two nonterminals in the rule.
Hence, all such situations cost $\Ocomp(\log M)$ per rule.

Lastly, it can happen that no nonterminal is removed from the rule,
even though a new common length is created:
this happens when in the $X_i$'s rule $X_i \to u X_i v X_j w$
both the $a$-suffix of $\eval(X_j)$ and the $a$-prefix of $\eval(X_k)$ 
are represented using the common lengths of $a$, moreover, $v \in a^*$.

Consider all such creations of powers in a fixed rule.
Let the consecutive letters, whose blocks were compressed, be
$a^{(1)}$, $a^{(2)}$, \ldots, $a^{(\ell)}$ and the corresponding blocks' lengths
$c_1$, $c_2$, \ldots, $c_\ell$.
Lastly, the $c_\ell$ repetitions of $a^{(\ell)}$ are replaced by $a^{(\ell+1)}$.
Observe, that $a^{(i+1)}$ does not need to be $a^{(i)}_{c_i}$,
as there might have been some other compression in between.

Recall the definition of \emph{weight}: for a letter it is the length of the represented
string in the original instance.
Consider the weight of the strings between $X_j$ and $X_k$.
Clearly, after the $i$-th blocks compression it is exactly $c_i \cdot \weight(a^{(i)})$.
We claim that $\weight(a^{(i+1)}) \geq c_i \weight(a^{(i)})$.

\begin{clm}
It holds that $\weight(a^{(i+1)}) \geq c_i \weight(a^{(i)})$.
\end{clm}
\begin{proof}
Right after the $i$-th blocks compression the string between $X_j$ and $X_k$
is simply $a^{(i)}_{c_i}$.
After some operations, this string consists of $c_{i+1}$ letters $a^{(i+1)}$.
All operations in \algFCPM{} do not remove the symbols from the string between two nonterminals in a rule
(removing of leading $a_R$ or ending $a_L$ from \mytext{} cannot affect letters between nonterminals).
Recall that we can think of the recompression as building of an SLP for the \pattern{} and \mytext.
In particular, one of the letters $a^{(i+1)}$ derives $a^{(i)}_{c_i}$,
Since in the derivation the weight is preserved, it holds that
$$
\weight(a^{(i+1)}) \geq \weight(a^{(i)}_{c_i}) = c_i \cdot \weight(a^{(i)}) \enspace .
$$
Which shows the claim.
\qed
\end{proof}
Thus $\weight(a^{(\ell)}) \geq \prod_{i=1}^{\ell-1} c_i$.
Still, by our assumption we consider only the cost of letters that appear in the pattern.
Hence, $a^{(\ell)}$ (or some heavier letter) appears in the pattern,
and so $M \geq \weight(a^{(\ell)})$ (note that this argument does not apply to $a^{(\ell+1)}$,
as it does not necessarily appear in \pattern).
Hence,
$$
\log(M) \geq \log \left(\prod_{i=1}^{\ell-1} c_i \right) =  \sum_{i=1}^{\ell-1} \log c_i.
$$
Taking into the account that $c_{\ell} \leq M$
(by the assumption we do not sort blocks of length greater than $M$ so they do not redirect any costs towards a rule),
the whole charge of $\sum_{i=1}^{\ell} \log c_i$ to the single rule is
in fact at most $2 \log M$.
\qed
\end{proof}
Summing over the rules gives the total cost of $\Ocomp((n+m)\log M)$, as claimed.

\subsubsection*{Large numbers}
When we estimated the running time of the \algremblocks, then we assumed that numbers up to $M$ can be manipulated in constant time.
We show that in fact this bound holds even if this assumption is lifted.
The only difference is that we cannot compare numbers in constant time.
However, if they are written as bit-vectors, the cost of each operation on a number $\ell$ is $\Theta(\log \ell)$.
For common lengths of letters appearing in \pattern{} and that are at most $M$ we estimated in Lemma~\ref{clm:cost from power to rule}
that such cost summed over all phases sums up to $\Ocomp((n+m)\log M)$.
So it is left to consider the cost for letters that do not appear in \pattern{}
and the cost for common lengths larger than $M$ of letters appearing in \pattern.

Concerning the letters not appearing in the pattern, we do not calculate their lengths at all,
so there is no additional cost.
For a common length $c > M$ of a letter from \pattern{} we spend $\Ocomp(\log M)$ time to find out that $c>M$.
Observe that if this common length is created because some $X_i$ generates it or during its creation a nonterminal
is removed from the rule, then this is fine as this happens only once per nonterminal/rule.

In the other case this common length appears between nonterminals in a rule for $X_i$.
Afterwards between those nonterminals there is a letter not appearing in \pattern.
Furthermore, compression cannot change it: in each consecutive phase there will be such a letter between
those nonterminals.
So there can be no more creation of common lengths of letters appearing in strings between those two nonterminals.
So the $\Ocomp(\log M)$ cost is charged to this rule only once.

\subsubsection*{Acknowledgements}
I would like to thank Pawe\l{} Gawrychowski for introducing me to the topic,
for pointing out the relevant literature~\cite{AlstrupBrodalRauhe,LifshitsMatching,LohreySLP,MehlhornSU97}
and discussions~\cite{GawrySLPEquality}.

\end{document}